\documentclass[preprint, tightenlines]{revtex4-1}
\usepackage{hyperref}
\usepackage{mathtools}
\usepackage{amsmath,mathrsfs}
\usepackage{graphicx}%
\usepackage{amsfonts}%
\usepackage{amssymb}
\usepackage{amsthm}
\usepackage{diagbox}

\newtheorem{theorem}{Theorem}
\newtheorem{definition}[theorem]{Definition}
\newtheorem{lemma}[theorem]{Lemma}
\newtheorem{proposition}[theorem]{Proposition}

\newcommand{\RNum}[1]{\uppercase\expandafter{\romannumeral #1\relax}}

\begin{document}

\title{Learning Nonlinear Input-Output Maps with Dissipative Quantum Systems}
\author{Jiayin \surname{Chen}}
\author{Hendra I. \surname{Nurdin}}
\email{h.nurdin@unsw.edu.au (corresponding author)}
\affiliation{School of Electrical Engineering and Telecommunications, The University of New South Wales (UNSW), Sydney NSW 2052, Australia}

\begin{abstract}
In this paper, we develop a theory of learning nonlinear input-output maps with fading memory by dissipative quantum systems, as a quantum counterpart of the theory of approximating such maps using classical dynamical systems. The theory identifies the properties required for a class of dissipative quantum systems to be {\em universal}, in that any input-output map with fading memory can be approximated arbitrarily closely by an element of this class.  We then introduce an example class of dissipative quantum systems that is provably universal. Numerical experiments illustrate that with a small number of qubits, this class can achieve comparable performance to classical learning schemes with a large number of tunable parameters. Further numerical analysis suggests that the exponentially increasing Hilbert space presents a potential  resource for dissipative quantum systems to surpass classical learning schemes for input-output maps. 
\end{abstract}

\maketitle

\section{Introduction}
\label{sec:introduction}
We are in the midst of the noisy intermediate-scale quantum (NISQ) technology era \cite{Preskill18}, marked by noisy quantum computers consisting of roughly tens to hundreds of  qubits. Currently there is a substantial interest in early applications of these machines that can accelerate the development of practical quantum computers, akin to how the humble hearing aid stimulated the development of integrated circuit (IC) technology \cite{Mills11}. NISQ quantum computing machines will not be equipped with quantum error correction and are thus incapable of performing continuous quantum computation. 

Several research directions are being explored for NISQ-class machines. One direction is to demonstrate so-called ``quantum supremacy'', in which NISQ machines can perform computational tasks that are demonstrably out of the reach of the most powerful digital supercomputers. The computational tasks include sampling problems such as boson sampling \cite{AA11,LBR17}, instantaneous quantum polynomial (IQP) computation \cite{BJS10,LBR17}, and sampling from random quantum circuits \cite{boixo2018characterizing}. Recent works have also proposed quantum machine learning algorithms that offer provable speedups over their classical counterparts \cite{biamonte2017quantum}. Another direction is the development of variational algorithms on hybrid classical-quantum machines to solve certain classes of optimization problems. Algorithms proposed include the quantum approximate optimization algorithm (QAOA) \cite{FGG14}, the quantum variational eigensolver (QVE) \cite{PMSYZLAO13,MRBA16} and variations and generalizations thereof, e.g., \cite{WHB18,mitarai2018quantum}. Experimental demonstration of QVE for calculating the ground-state energy of small molecules has been reported in \cite{KMTTBCG17}, while the application of QAOA for unsupervised learning of a clustering problem can be found in \cite{Otterbach17}.  

An alternative paradigm to the quantum gate-based approaches above is to harness the computational capability of dissipative quantum systems. Dissipative quantum dynamics has been shown to be able to realize universal quantum computation \cite{verstraete2009quantum} and has been applied in a time-delay fashion for supervised quantum machine learning without intermediate measurements \cite{alvarez2017supervised}. Recently, quantum reservoir computers (QRCs) are introduced to harness the complex real-time quantum dissipative dynamics \cite{fujii2017harnessing,nakajima2019boosting}. This approach is essentially a quantum implementation of classical \textit{reservoir computing} schemes, in which a dynamical system processes an input sequence and produces an output sequence that approximates a target sequence, see, e.g.,  \cite{JH04,MNM02,LJ09}. The main philosophy in reservoir computing is that the dynamics in arbitrary naturally occurring or engineered  dynamical systems could potentially be exploited for computational purposes. In particular, a dynamical system could be used for computation without precise tuning or optimization of its parameters.  To possess temporal information,  the systems are required to satisfy three properties \cite{MNM02}:  the \textit{convergence property} \cite{PWN05}, the \textit{fading memory property} \cite{BC85} and form a family of systems with the \textit{separation property}. The convergence property ensures that computations performed by a dynamical system are independent of its initial condition, and the fading memory property implies that outputs of a dynamical system stay close if the corresponding inputs are close in recent times. The separation property states that there should be a member in the family of systems with dynamics sufficiently rich to distinguish any two different input sequences. Classical reservoir computing has been realized as simple nonlinear photonic circuits with a delay line \cite{Appel11} and in neuromorphic computing based on nanoscale oscillators \cite{Torre17}, and it has been demonstrated to achieve state-of-the-art performance on applications such spoken digit recognition \cite{Torre17}.

Nonlinear input-output (I/O) maps with fading memory can be approximated by a series expansion such as the well-known Volterra series \cite{BC85}. They can also be approximated by a family of classical nonlinear dynamical systems that have the three properties introduced in the previous paragraph. Such a family of dynamical systems is said to be {\em universal} (or possesses the universality property) for nonlinear I/O maps with fading memory. They  include various classical reservoir computing schemes such as liquid state machines \cite{MNM02}, echo-state networks (ESNs) \cite{grigoryeva2018echo}, linear reservoirs with polynomial readouts (LRPO), and non-homogeneous state-affine systems (SAS) \cite{grigoryeva2018universal}. However, a theoretical framework for the learning of nonlinear fading memory I/O maps by quantum systems is so far lacking. Moreover, an extended investigation into the potential advantage quantum systems offer over classical reservoir computing schemes has not been conducted. The provision of such a learning theory, the demonstration of a class of quantum model that is provably universal, and a study of this model via numerical experiments are the main contributions of this paper.

The paper is organized as follows. In Sec.~\ref{sec:fading_memory_maps}, we formally define fading memory maps. In Sec.~\ref{sec:learning}, we formulate the theory of learning nonlinear fading memory maps with dissipative quantum systems. Sec.~\ref{sec:universal_class} introduces a concrete universal class of dissipative quantum systems. Sec.~\ref{sec:numerical} numerically demonstrates the emulation performance of the proposed universal class in the absence and presence of decoherence. The effect of different input encodings on the learning capability of this class is investigated. An in-depth comparison between this universal class and ESNs is also conducted. We conclude this section by discussing the potential of this universal class to surpass classical schemes when implemented on a NISQ machine. In Sec.~\ref{sec:discussion}, we discuss the feasibility of proof-of-principle experiments of the proposed scheme on existing NISQ machines. Detailed results and numerical settings are collected in and can be found in the Appendix.

\section{Fading memory maps} \label{sec:fading_memory_maps}
Let $\mathbb{Z}$ denote the set of all integers and $\mathbb{Z}^{-} = \{\ldots, -1, 0\}$. Let $u=\{\ldots,u_{-1},u_{0},u_{1},\ldots\}$ be a real bounded input sequence with $\sup_{k \in \mathbb{Z}} |u_k|<\infty$. We say that a real output sequence $y = \{\ldots, y_{-1},y_{0},y_{1},\ldots\} $ is related to $u$ by a time-invariant causal map $M$ if $y_k = M(u)_k = M(\tilde u_{\ell})_k$ for any integer $\ell$, any $k \leq \ell$,  and any sequence $\tilde{u}_{\ell}$  such that $\left. \tilde{u}_{\ell} \right|_{\ell} =  \left. u \right|_{\ell}$.
Here, $M(u)_k$ denotes the output sequence at time $k$ given the input sequence $u$, and $\left. u\right|_k = \{\ldots, u_{k-2}, u_{k-1}, u_{k} \}$ is the input sequence $u$ truncated after time $k$.

For a fixed real positive constant $L$ and a compact subset $D \subseteq \mathbb{R}$, we are interested in the set $K_{L}(D)$ consisting of input sequences such that for all $k \in \mathbb{Z}$, $u_{k} \in D \cap [-L, L]$. We say a time-invariant causal map $M$ defined on $K_{L}(D)$ has the fading memory property with respect to a decreasing sequence $w = \{w_{k}\}_{k \geq 0}$, $\lim_{k \rightarrow \infty} w_{k} = 0$ if, for any two input sequences $u$ and $v$, $|M(u)_0-M( v)_0| \rightarrow 0$ whenever $\sup_{k \in  \mathbb{Z}^-} |w_{-k}(u_k-v_k)| \rightarrow 0$. In other words, if the elements of two sequences agree closely up to some recent past before  $k=0$, then their output sequences will also be close at $k=0$.

\section{Learning nonlinear fading memory maps with dissipative quantum systems}\label{sec:learning}
Since fading memory maps are time-invariant, any dynamical system that is used to approximate them must forget its initial condition. Classical dynamical systems with this property are referred to as {\em convergent systems} in control theory \cite{PWN05}, and the property is known as the {\em echo state property} in the context of ESNs  \cite{JH04,BY06}. For dissipative quantum systems, this means that for the same input sequence, density operators asymptotically converge to the same sequence of density operators, independently of their initial values. We emphasize that the dissipative nature of the quantum system is {\em essential} for the learning task. Without it the system clearly cannot be convergent. 
 
Consider a quantum system consisting of $n$ qubits with a Hilbert space $\mathbb{C}^{2^n}$ of dimension $2^n$ undergoing the following discrete-time dissipative evolution:
\begin{eqnarray}
\rho_{k}  &=& T(u_{k}) \rho_{k-1}, \label{eq:rho-dynamics}
\end{eqnarray}
for $k=1,2,\ldots,$ with initial condition $\rho(0)=\rho_0$. Here, $\rho_k=\rho(k\tau)$ is the system density operator at time $t=k\tau$ and $\tau$ is a (fixed) sampling time, and $T(u_k)$ is a completely positive trace preserving (CPTP) map for each $u_k$. In this setting, the real input sequence $\{u_1, u_2,\ldots\}$ determines the system's evolution. The overall input-output map in the long time limit is in general non-linear. Let $\|\cdot\|_{p}$ denote any Schatten $p$-norm for $p \in [1, \infty)$ defined as $\|A\|_{p} = {\rm Tr}(\sqrt{A^* A}^{p})^{1/p}$, where $A$ is a complex matrix and $*$ is the conjugate transpose operator. In Appendix~[\ref{app:covergence}, Theorem~\ref{theorem:convergence}] , we show that if for all $u_{k} \in D \cap [-L, L]$, the CPTP map $T(u_k)$ restricted on the hyperplane $H_{0}(2^n)$ of $2^n \times 2^n$ traceless Hermitian operators satisfies $\|T(u_k)\rvert_{H_{0}(2^n)}\|_{2-2} \coloneqq \sup_{A \in H_{0}(2^n), A \neq 0}\frac{\|T(u_k)A\|_{2}}{\|A\|_{2}} \leq 1 - \epsilon$ for some $0 < \epsilon \leq 1$, then under any input sequence $u \in K_{L}(D)$, it will forget its initial condition and is therefore convergent. This means that for any two initial density operators $\rho_{j,0}$ ($j=1,2$) and the corresponding density operators $\rho_{j,k}$ at time $t=k\tau$, we will have that
$$
\mathop{\lim}_{k \rightarrow \infty }\| \rho_{1,k} -\rho_{2,k}\|_{2} = \mathop{\lim}_{k \rightarrow \infty}\left\|\left( \overleftarrow{\prod}_{j=1}^{k} T(u_j) \right) (\rho_{1,0}-\rho_{2,0}) \right\|_{2} =0,
$$
where $\overleftarrow{\prod}_{j=1}^{k}$ is a time-ordered composition of maps $T(u_{j})$ from right to left. 

Let $\mathcal{D}(\mathbb{C}^{2^n})$ denote the convex set of all density operators on $\mathbb{C}^{2^n}$. We introduce an output sequence $\bar{y}$ in the form 
\begin{equation}
\bar{y}_k = h(\rho_k), \label{eq:y-dynamics}
\end{equation}
where $h: \mathcal{D}(\mathbb{C}^{2^n}) \rightarrow \mathbb{R}$ is a real functional of $\rho_k$. Eqs. \eqref{eq:rho-dynamics} and \eqref{eq:y-dynamics} define a quantum dynamical system with input sequence $u$ and output sequence $\bar{y}$. We now require the separation property. Consider a family  $\mathcal{F}$ of distinct quantum systems described by Eqs. \eqref{eq:rho-dynamics} and \eqref{eq:y-dynamics}, but possibly having differing number of qubits. Let $u$ and $u'$ be two input sequences in $K_L(D)$ that are not identical, $u_k \neq u'_k$ for at least one $k$, and let $\bar{y}$ and $\bar{y}'$ be the respective outputs of the quantum system for these inputs. We say that the family $\mathcal{F}$ is {\em separating} if for any non-identical inputs $u$ and $u'$ in $K_L(D)$, there exists a member in this family with non-identical outputs $\bar{y}$ and $\bar{y}'$.  As stated in Appendix~[\ref{app:universality}, Theorem~\ref{theorem:universal_dissipative}], any family of convergent dissipative quantum systems that implement  fading memory maps with the separation property, and which forms an algebra of maps containing the constant maps, is universal and can approximate any I/O map with fading memory arbitrarily closely.  

\section{A universal class of dissipative quantum systems} \label{sec:universal_class}
We now specify a class of dissipative quantum systems that is provably universal in approximating fading memory maps defined on $K_{1}([0,1])$.
The class consists of systems that are made up of $N$ {\em non-interacting} subsystems initialized in a product state of the $N$ subsystems, with subsystem $K$ consisting of $n_{K}+1$ qubits, $n_K$ ``system'' qubits and a single ``ancilla" qubit. We label the qubits of subsystem $K$ by an index $i^K_j$ that runs from $j=0$ to $j=n_K$, with $i_0^K$ labeling the ancilla qubit. The $n_K+1$ qubits interact via the Hamiltonian 
$$H_K = \sum_{j_1=0}^{n_K} \sum_{j_2 = j_1+1}^{n_K} J_K^{j_1,j_2} (X^{(i^K_{j_1})}X^{(i^K_{j_2})} + Y^{(i^K_{j_1})}Y^{(i^K_{j_2})}) + \sum_{j=0}^{n_K} \alpha Z^{(i_{j}^{K})},$$
where  $J_K^{j_1,j_2}$ and $\alpha$ are real-valued constants, while $X^{(i^K_{j})}$, $Y^{(i^K_{j})}$ and $Z^{(i^K_j)}$ are Pauli $X$, $Y$ and $Z$ operators of qubit $i^K_j$. The ancilla qubits for all subsystems are periodically reset at time $t =k\tau$ and prepared in the input-dependent mixed state $\rho^K_{i_{0,k}} = u_k |0 \rangle \langle 0| + (1-u_k) |1 \rangle \langle 1|$ (with $0 \leq u_k \leq 1$). The system qubits are initialized at time $t=0$ to some density operator.  The density operator $\rho^K_k$ of the $K^{\rm th}$ subsystem qubits evolves during time $(k-1)\tau < t < k\tau$ according to $\rho^K_k = T_{K}(u_k) \rho^K_{k-1}$, where  $T_{K}(u_k)$ is the CPTP map defined by $
T_{K}(u_k) \rho^K_{k-1} = {\rm Tr}_{i^K_0} \left(e^{- i H_K \tau }\rho^K_{k-1} \otimes \rho^{K}_{i_{0,k}} e^{i H_K \tau }\right)$ and ${\rm Tr}_{i^K_0}$  denotes the partial trace over the ancilla qubit of subsystem $K$.  We now specify an output functional $h$ associated with this system.  We will use a single index to label the system qubits from the $N$ subsystems, the ancilla qubits are not used in the output. Consider an individual system qubit with index $j$, with $j$ running from 1 until $n = \sum_{K=1}^N n_K$. The output functional $h$  is defined to be of the general form,

\begin{equation}
\label{eq:mv-out}
\bar{y}_{k} = h(\rho_k)= C + \sum_{d=1}^{R} \sum_{i_{1} = 1}^{n} \sum_{i_{2} = i_{1} + 1}^{n} \cdots \sum_{i_{n} = i_{n-1}+1}^{n} \sum_{r_{i_1} + \cdots + r_{i_n} = d} w_{i_1, \ldots, i_{n}}^{r_{i_1}, \ldots, r_{i_n}} \langle Z^{(i_{1})} \rangle_{k}^{r_{i_1}} \cdots \langle Z^{(i_n)} \rangle_{k}^{r_{i_n}}
\end{equation}
where $C$ is a constant, $R$ is an integer and $\langle Z^{(i)} \rangle_k = {\rm Tr}(\rho_k Z^{(i)})$ is the expectation of the operator $Z^{(i)}$. We note that the functional $h$ (the right hand side of the above) is a multivariate polynomial in the variables $\langle Z^{(i)} \rangle_k$ ($i=1,\ldots,n$) and these expectation values depend on input sequence $u = \{u_k\}$. Thus computing $\bar{y}_k$ only involves estimating the expectations $\langle Z^{(i)}\rangle_k$ and the degree of the polynomial $R$ can be chosen as desired. If $R=1$ then $\bar{y}_k$ is a simple linear function of the expectations. 

This family of dissipative quantum systems exhibits two important properties, see Appendix~\ref{app:FMP_algebra} and \ref{app:universal_class} for the proofs. Firstly, if for each subsystem $K$ with $n_{K}$ qubits and for all $u_{k} \in [0, 1]$, $\|T_{K}(u_k)\rvert_{H_{0}(2^{n_{K}})}\|_{2-2} \leq 1-\epsilon_{K}$ for some $0 < \epsilon_{K} \leq 1$, then this family forms a polynomial algebra consisting of systems that implement fading memory maps. Secondly, a convergent single-qubit system with a linear output combination of expectation values (ie. $n=1$, $N=1$ and $R=1$), separates points of $K_{1}([0, 1])$. These two properties and an application of the Stone-Weierstrass Theorem \cite[Theorem 7.3.1]{Dieudonne13} guarantee the universality property.

The class specified above is a variant of the QRC model in \cite{fujii2017harnessing} but is provably universal by the theory of the previous section. The differences are in the general form of the output and, in our model, the ancilla qubit is not used in computing the output. Also, we do not consider time-multiplexing. We remark that time-multiplexing can be in principle incorporated in the model using the same theory. However, this extension is more technical and will be pursued elsewhere.

\section{Numerical experiments} \label{sec:numerical}
We demonstrate the emulation performance of the universal class introduced above in learning a number of benchmarking tasks. A random input sequence $u^{(r)} = \{u^{(r)}_{k}\}_{k > 0}$, where each $u_{k}^{(r)}$ is randomly uniformly chosen from $[0, 0.2]$, is applied to all computational tasks. We apply the multitasking method, in which we simulate the evolution of the quantum systems and record the expectations $\langle Z^{(i)} \rangle_{k}$ for all timesteps $k$ once, while the output weights $C$ and $w_{i_1, \ldots, i_{n}}^{r_{i_1}, \ldots, r_{i_n}}$ in Eq.~\eqref{eq:mv-out} are optimized independently for each computational task.

The linear reservoirs with polynomial outputs (LRPO) implement a fading memory map, whose discrete-time dynamics is of the form \cite{grigoryeva2018universal, BC85},
\begin{equation*}
\begin{cases}
x_{k} = Ax_{k-1} + cu_{k} \\
y_{k} = \hat{h}(x_{k}),
\end{cases}
\end{equation*}
where we choose $c \in \mathbb{R}^{1400}$ with elements randomly uniformly chosen from $[0, 4]$ and $\hat{h}$ to be a degree two multivariate polynomial, whose coefficients are randomly uniformly chosen from $[-0.1, 0.1]$. We choose $A$ to be a diagonal block matrix $A = {\rm diag}(A_{1}, A_{2}, A_{3})$, where $A_{1}, A_{2}$ and $A_{3}$ are $200 \times 200$, $500 \times 500$ and $700 \times 700$ real matrices, respectively. Elements of $A_{i}$ $(i=1,2,3)$ are randomly uniformly chosen from $[0, 4]$. To ensure the convergence and the fading memory property, the maximum singular value of each $A_{i}$ is randomly uniformly set to be $\sigma_{\max}(A_{i}) < 1$ \cite{grigoryeva2018universal}. In this setting, each linear reservoir defined by $A_{i}$ evolves independently, while the output of the LRPO depends on all state elements $x_{k} \in \mathbb{R}^{1400}$.

It is interesting to investigate the performance of the universal class in learning tasks that do not strictly implement fading memory maps as defined here. We apply the universal class to approximate the outputs of a missile moving with a constant velocity in the horizontal plane \cite{ni1996new} and the nonlinear autoregressive moving average (NARMA) models \cite{atiya2000new}. The nonlinear dynamics of the missile is given by
\begin{equation*}
\begin{cases}
\dot{x}_{1} = x_{2} - 0.1 \cos(x_1) (5x_1 - 4 x_1^3 + x_1^5) - 0.5 \cos(x_1)\tilde{u} \\
\dot{x}_{2} = -65 x_1 + 50 x_1^3 - 15x_1^5 - x_2 - 100 \tilde{u}
\end{cases}
\end{equation*}
where $y = x_{2}$ is the output. We make a change of variable of the input $\tilde{u}=5u - 0.5$ so that the input range is the same as in \cite{ni1996new}. The missile dynamics is simulated by the Runge-Kutta $(4, 5)$ formula implemented by the ode45 function in MATLAB \cite{dormand1980family}, with a sampling time of $4 \times 10^{-4}$ seconds for a time span of $1$ second, subject to the initial condition $\begin{pmatrix} x_1 & x_2 \end{pmatrix}^{T} = \begin{pmatrix} 0 & 0 \end{pmatrix}^{T}$. We denote this task as Missile. The NARMA models are often used to benchmark algorithms for learning time-series. The outputs of each NARMA model depend on its time-lagged outputs and inputs, specified by a delay $\tau_{\rm NARMA}$. We denote the corresponding task to be NARMA$\tau_{\rm NARMA}$.

We focus on members of the universal class with a single subsystem ($N=1$) and a small number of system qubits $n = \{2, 3, 4, 5, 6\}$, and denote this subset of the universal class as SA. We will drop the subsystem index $K$ from now on. For all numerical experiments, the parameters of SA are chosen as follows. We introduce a scale $S>0$ such that the Hamiltonian parameters $J^{j_1, j_2}/S$, $\alpha/S = 0.5$ and $\tau S = 1$ are dimensionless. As for the QRCs in \cite{fujii2017harnessing}, we randomly uniformly generate $J^{j_1, j_2}/S$ from $[-1, 1]$ and, to ensure convergence, select the resulting Hamiltonians for experiments if the associated CPTP map is convergent. We numerically test the convergence property by checking if $50$ randomly generated initial density operators converge to the same density operator in $500$ timesteps under the input sequence $u^{(r)}$.

Each numerical experiment firstly washouts the effect of initial conditions of SA and all target maps with $500$ timesteps. This is followed by a training stage of $1000$ timesteps, where we optimize the output weights $C$ and $w_{i_1, \ldots, i_{n}}^{r_{i_1}, \ldots, r_{i_n}}$ of SA by standard least squares to minimize the error $\sum_{k=501}^{1500} |y_k - \bar{y}_{k}|^2$ between the target output sequence $y$. In practical implementation, computation of the expectations $\langle Z^{(j)} \rangle_{k}$ is offloaded to the quantum subsystem, and only a simple classical processing method is needed to optimize the output weights. For this reason, we associate the output weights $C$ and $w_{i_1, \ldots, i_{n}}^{r_{i_1}, \ldots, r_{i_n}}$ in Eq.~\eqref{eq:mv-out} with  \textit{(classical) computational nodes}, with the number of such nodes being equal to the number of output weights. While the number of computational nodes for SA can be chosen arbitrarily by varying the degree $R$ in the output, the state-space `size' of the quantum system is $2^{n} (2^{n}+1)-2^{n} = 4^{n}$. This state-space size corresponds to the number of real variables needed to describe the evolution of elements of the system density operator. Note that since the density operator has unity trace, only up to at most $4^n-1$ of these nodes are linearly independent.

On the other hand, for ESNs \cite{JH04}, the number of computational nodes and the state-space size always differ by one (i.e. by the tunable constant output weight). For an ESN with $m$ reservoir nodes (E$m$), the number of computational nodes is $m+1$ and its state-space size is $m$. For the Volterra series \cite{BC85} with kernel order $o$ and memory $p$ (V$o,p$), the number of computational nodes is $(\frac{p^{o+1} - p}{p-1}+1)$. We select $m$ and $(o,p)$ such that the number of computational nodes is at most $801$. This reduces the chance of overfitting for learning a sequence of length $1000$ \cite{lukovsevivcius2012practical}. For detailed numerical settings for ESNs and the Volterra series, see Appendix~\ref{app:numerical}. We analyze the performance of all learning schemes during an evaluation phase consisting of $1000$ timesteps, using the normalized mean-squared error ${\rm NMSE} \coloneqq \sum_{k=1501}^{2500}|\bar{y}_{k} - y_{k}|^2/\sum_{k=1501}^{2500} |y_k - \frac{1}{1000}\sum_{k=1501}^{2500} y_{k}|^2$, where $y$ is the target output and $\bar{y}$ is the approximated output. For each task and each $n$, NMSEs of 100 convergent SA samples are averaged for analysis.

\begin{figure}[ht!]
\centering
\includegraphics[trim={15mm 10mm 0 5mm}, clip=true, scale=0.5]{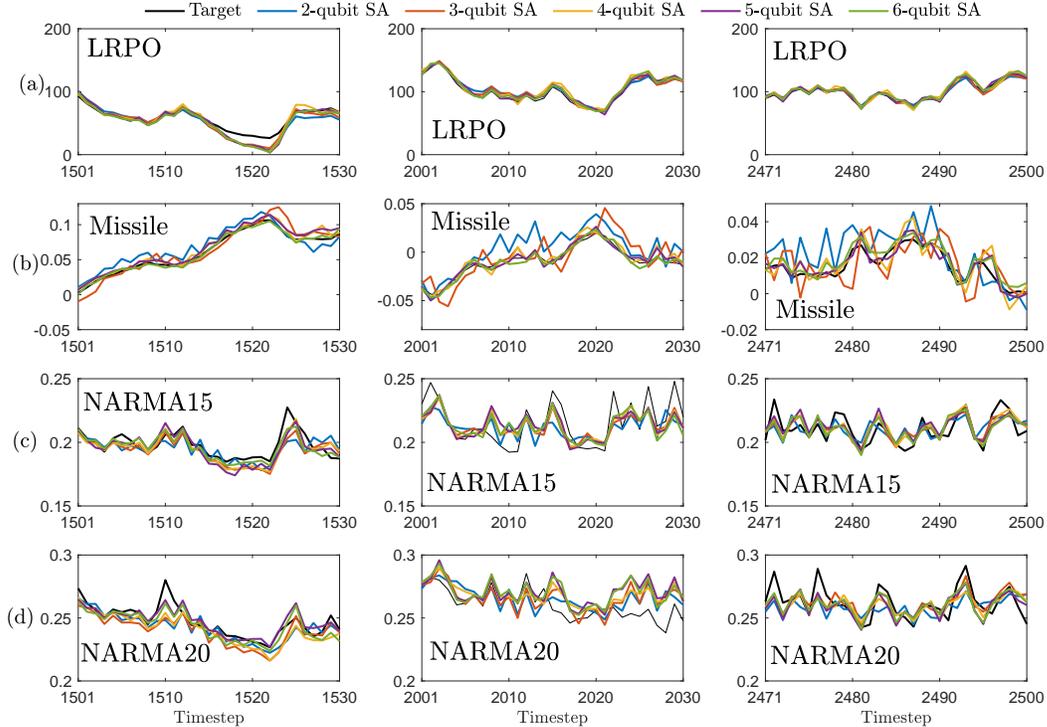}
\caption{Typical SA outputs during the evaluation phase, for the (a) LRPO, (b) Missile (c) NARMA15 and (d) NARMA20 tasks. The leftmost, middle and rightmost panels show the outputs for timesteps 1501-1530, 2001-2030 and 2471-2500, respectively}
\label{figure:trace}
\end{figure}

\begin{figure}
\centering
\includegraphics[trim={38mm 0mm 35mm 5mm}, clip=true, scale=0.45]{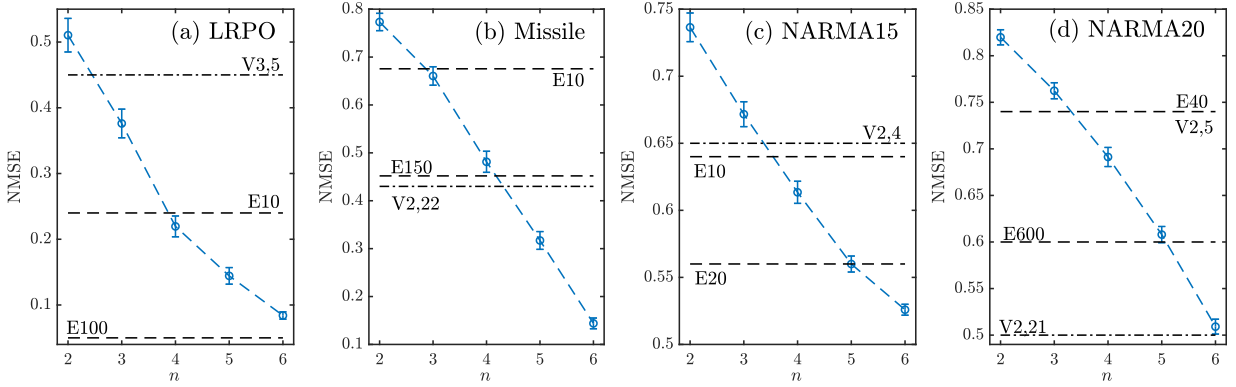}
\caption{Average SA NMSE for the (a) LRPO, (b) Missile, (c) NARMA15 and (d) NARMA20 tasks, the error bars represent the standard error. For comparison, horizontal dashed lines labeled with ``E$m$'' indicate the average performance of ESNs with $m$ computational nodes, and horizontal dot-dashed lines labeled with ``V$o,p$'' indicates the performance of Volterra series with kernel order $o$ and memory $p$. Overlapping dashed and dot-dashed lines are represented as dashed lines}
\label{figure:overview}
\end{figure}

\begin{figure}[ht!]
\centering
\includegraphics[trim={40mm 20mm 35mm 10mm}, clip=true, scale=0.45]{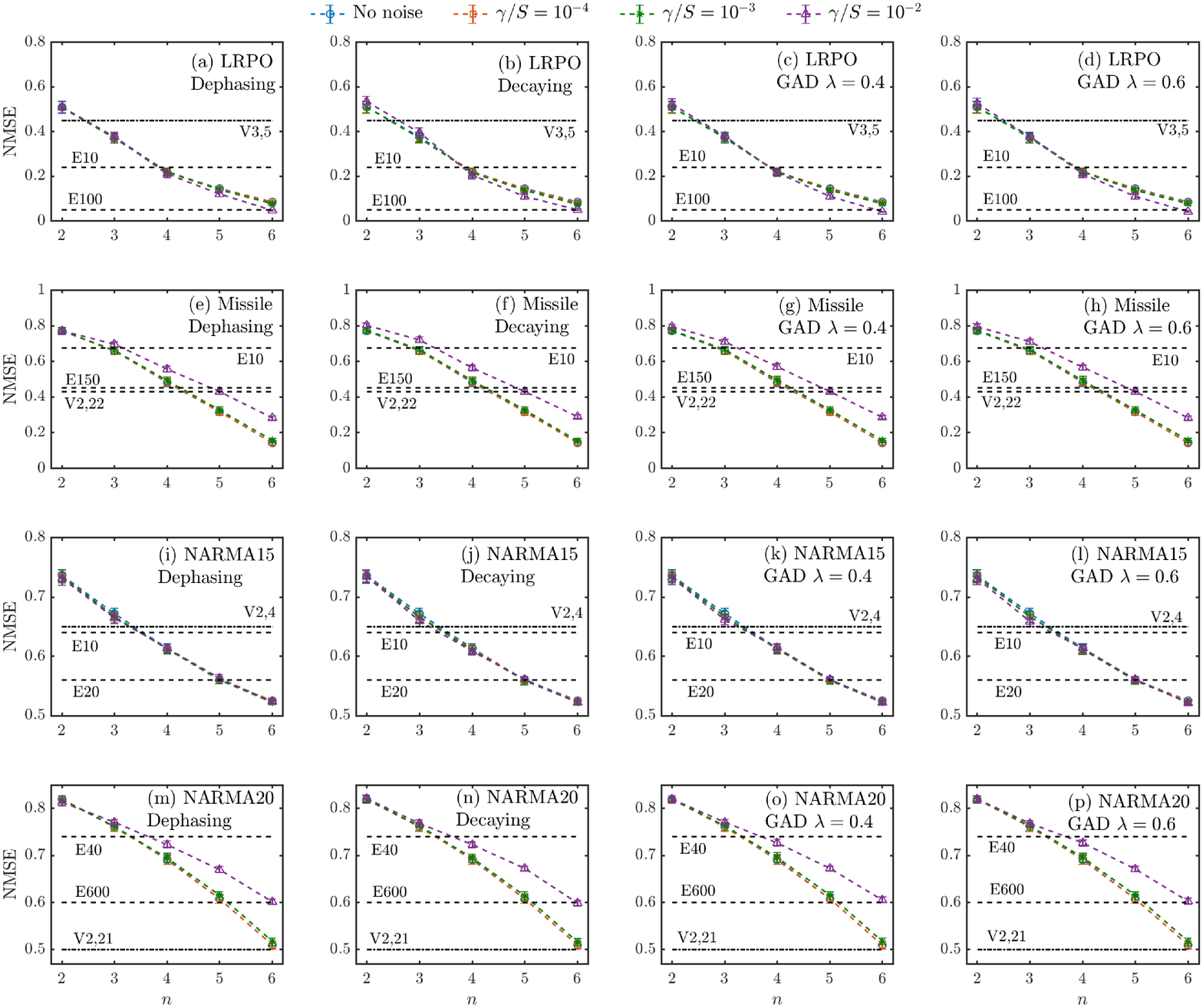}
\caption{Average SA NMSE for the LRPO, Missile, NARMA15 and NARMA20 tasks under decoherence. For comparison, the average SA NMSE without the effect of noise is also plotted. In all plots, the error bars represent the standard error}
\label{figure:noise}
\end{figure}

\subsection{Overview of SA learning performance} \label{subsec:overview}
We present an overview of SA performance in learning the LRPO, Missile, NARMA15 and NARMA20 tasks. The degree of the multivariate polynomial output Eq.~\eqref{eq:mv-out} is fixed to be $R=1$, so that the number of computational is $n+1$ for each $n$. Fig.~\ref{figure:trace} shows the typical SA outputs for the LRPO, Missile, NARMA15 and NARMA20 tasks during the evaluation phase. It is observed that the SA outputs follow the LRPO outputs closely, while SA is able to approximate the Missile and NARMA tasks relatively closely. For all computational tasks, as the number of system qubits $n$ increases, the SA outputs better approximate the target outputs. This is quantitatively demonstrated in Fig.~\ref{figure:overview}, which plots the average SA NMSE as $n$ increases.

From Fig.~\ref{figure:overview} we can see that the SA model with a small number of computational nodes performs comparably as ESNs and the Volterra series with a large number of computational nodes. For example, the average NMSE of $6$-qubit SA with $7$ computational nodes is comparable to the average NMSE of E$100$ with $101$ computational nodes in the LRPO task. On average, $5$-qubit SA with $6$ computational nodes performs better than V$2,22$ with 504 computational nodes in the Missile task. In the NARMA15 task, $4$-qubit SA with $5$ computational nodes outperforms V$2,4$ with $21$ computational nodes. In the NARMA20 task, $5$-qubit SA performs comparably as E600. Our results are similar to the performance of the QRCs with time multiplexing reported in \cite{fujii2017harnessing}, where the QRCs are demonstrated to perform comparably as ESNs with a larger number of trainable computational nodes. However, for the small number of qubits investigated, the rate of decrease in the average NMSE is approximately linear despite the dimension of the Hilbert space increases exponentially as $n$ increases. For both the NARMA tasks, the average NMSEs for 2-qubit and 6-qubit SA are of the same order of magnitude.  A larger number of additional system qubits are required to substantially reduce the SA task error.

\subsection{SA performance under decoherence}
\label{subsec:SA_decoherence}
We further validate the feasibility of the SA model in the presence of the dephasing, decaying noise and the generalized amplitude damping (GAD) channel. We simulate the noise by applying the Trotter-Suzuki formula \cite{trotter1959product,suzuki1971relationship}, in which we divide the normalized time interval $\tau S=1$ into $50$ small time intervals $\delta_{t} = \tau S/50$, and alternatively apply the unitary interaction and the Kraus operators $\{M_{l}^{(j)}(\frac{\gamma}{S})\}_{l}$ of each noise type, each for a time duration of $\delta_{t}$. Each of the $l$-th Kraus operator is applied for all system and ancilla qubits $j= 1, \ldots, n+1$, and $\gamma/S$ denotes the noise strength. For all noise types, we apply the same noise strengths $\gamma/S = \{10^{-4}, 10^{-3}, 10^{-2}\}$, which are within the experimentally feasible range for systems like NMR ensembles \cite{vandersypen2001experimental} and some current superconducting NISQ machines \cite{IBMQ20}. Under the dephasing noise, the density operator $\rho$ of the system and ancilla qubits evolves according to $\rho \rightarrow \frac{1 + e^{-2 \frac{\gamma}{S} \delta_{t}}}{2} \rho + \frac{1 - e^{-2 \frac{\gamma}{S} \delta_{t} }}{2} Z^{(j)} \rho Z^{(j)}$, such that the diagonal elements in $\rho$ remain invariant while the off-diagonal elements decay. The GAD channel gives rise to the evolution $\rho \rightarrow \sum_{l=0}^{3} M_{l}^{(j)}(\frac{\gamma}{S},\lambda) \rho (M_{l}^{(j)}(\frac{\gamma}{S},\lambda))^{\dagger}$, where $\dagger$ denotes the adjoint, and the Kraus operators $M_{l}^{(j)}(\frac{\gamma}{S}, \lambda)$ $(l=0,1,2,3)$ depend on an additional finite temperature parameter $\lambda \in [0, 1]$ \cite{nielsen2002quantum}. When $\lambda=1$, we recover the amplitude damping channel (the decaying noise), which takes a mixed state into the pure ground state $|0\rangle \langle 0|$ in the long time limit. For $\lambda \neq 1$, we investigate the SA task performance under the GAD channel for $\lambda = \{0.2, 0.4, 0.6, 0.8\}$. The GAD channel affects both the diagonal and off-diagonal elements of the density operator.

Fig.~\ref{figure:noise} plots the average SA NMSE under the dephasing, decaying and GAD with $\lambda=\{0.4, 0.6\}$ for all noise strengths. See Appendix~\ref{app:decoherence} for the average NMSE under the GAD channel for all chosen temperature parameters. Fig.~\ref{figure:noise} indicates that for the same noise strength, different noise types affect the SA task performance in a similar manner. For noise strengths $\gamma/S = \{10^{-4}, 10^{-3}\}$, all noise types do not significantly degrade SA task performance for the computational tasks. However, the impact of the noise strength $\gamma/S = 10^{-2}$ is more pronounced, particularly for a larger number of system qubits.

\begin{figure}[!htb]
\centering
\includegraphics[trim={40mm 5mm 15mm 10mm}, clip, scale=0.45]{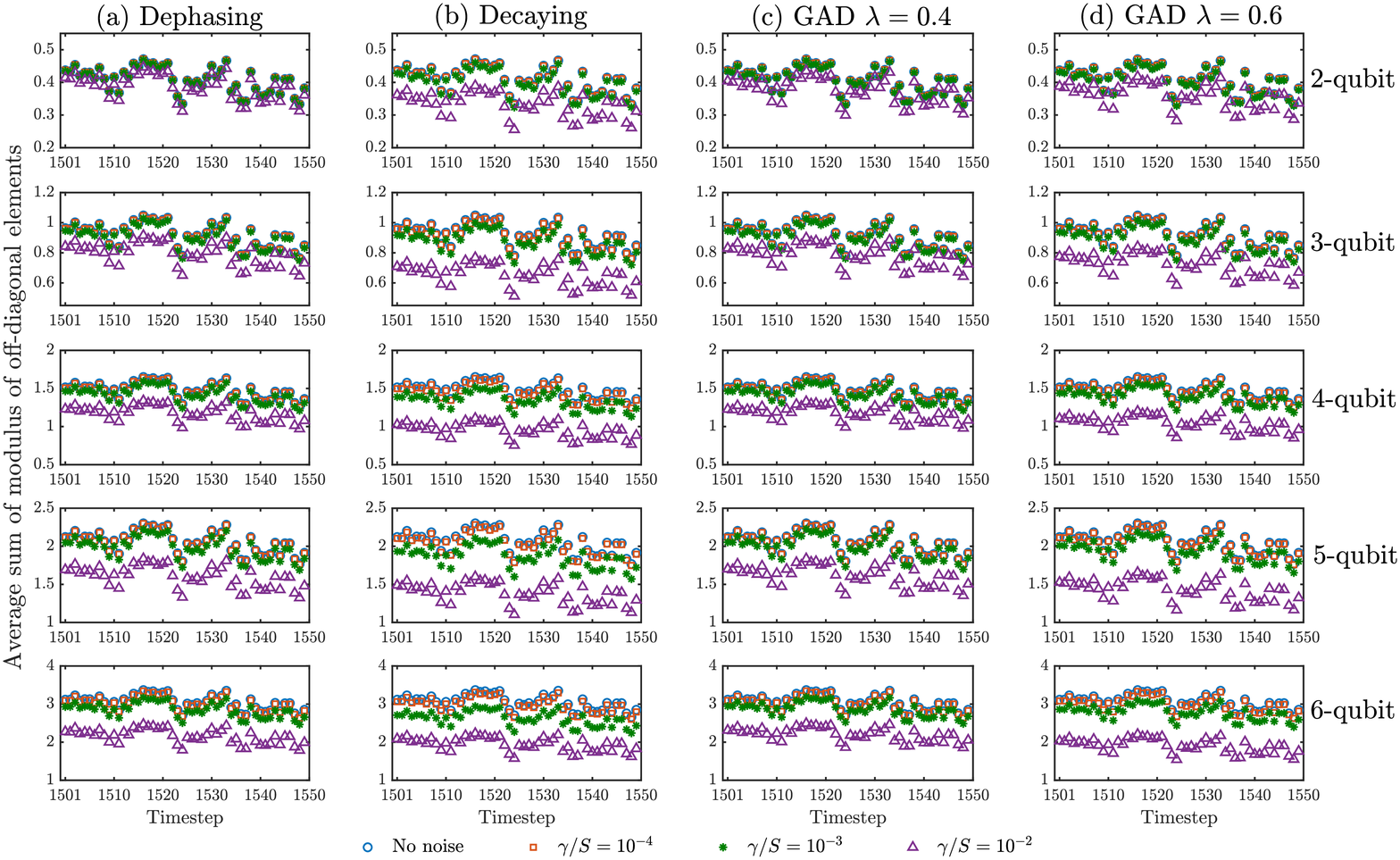}
\caption{Average sum of complex modulus of off-diagonal elements in the system density operator for timesteps 1501-1550, under the (a) dephasing noise, (b) decaying noise, (c) GAD with $\lambda=0.4$ and (d) GAD with $\lambda=0.6$. Row $n-1$ in the figure corresponds to the average sum for $n$-qubit SA}
\label{fig:eof}
\end{figure}

Changes in the SA task error under the effect of the decaying noise and the GAD channel are anticipated, since the expectations $\langle Z^{(j)} \rangle_{k}$ in the output depend on the diagonal elements of the system density operator, which are affected by both of these noise types. However, the SA task performance is also affected by the dephasing noise, which does not change the diagonal elements. A possible explanation for this behavior is a loss of degrees of freedom, in the sense that off-diagonal elements of the density operator become smaller and the density operator looks more like a classical probability distribution. Alternatively, this could be viewed as the off-diagonal elements contributing less to the overall computation. To support this explanation, for the dephasing, decaying and the GAD with $\lambda=\{0.4, 0.6\}$, and for each $n$, we sum the complex modulus of off-diagonal elements in the system density operator for the 100 $n$-qubit SA samples simulated above. The average of these 100 sums is plotted for the first 50 timesteps during the evaluation phase in Fig.~\ref{fig:eof}. That is, Fig.~\ref{fig:eof} plots $\frac{2}{n_s} \sum_{l=1}^{n_s} \sum_{r=1}^{2^n} \sum_{s=r+1}^{2^n} |\rho_{k,rs}^{(l)}|$, where $n_s = 100$ is the number of different random SA samples. Here $\rho_{k,rs}^{(l)}$ denotes the element of $\rho_{k}^{(l)}$ in row $r$ and column $s$ (the superscript $(l)$ indexing the SA sample).

Fig.~\ref{fig:eof} shows that as the noise strength increases, the average sum decreases, particularly with the noise strength $\gamma/S = 10^{-2}$. Similar trends are observed for the GAD channel for all the temperature parameters chosen, and the observed trend for the average sum persists as the timestep increases to 2500 (see Appendix~\ref{app:decoherence}). The results presented in Fig.~\ref{fig:eof} further indicate that though the output of SA depends solely on the diagonal elements of the density operator, nonzero off-diagonal elements are crucial for improving the SA emulation performance. This provides a plausible explanation for the improved performance achieved by increasing the number of qubits, thereby increasing Hilbert space size and the number of non-zero off-diagonal elements. Further investigation into this topic is presented in Sec.~\ref{subsec:further_comparison}.

\subsection{Effect of different input encodings}
Our proposed universal class encodes the input $u_{k} \in [0, 1]$ into the mixed state $\rho_{i_{0}, k} = u_{k}| 0 \rangle \langle 0 | + (1-u_{k}) | 1 \rangle \langle 1 |$. Other input encoding possibilities include
the pure state $\rho_{i_{0}, k} = (\sqrt{u_{k}} | 0 \rangle + \sqrt{1-u_{k}} |1 \rangle) (\sqrt{u_{k}} \langle 0 | + \sqrt{1-u_{k}} \langle 1 |)$ used in the QRC model \cite{fujii2017harnessing}, encoding the input into the phase $\rho_{i_{0}, k} = \frac{1}{2} (|0 \rangle + e^{-i u_{k}} |1 \rangle)(\langle 0 | + e^{i u_{k}} \langle 1 |)$, and encoding the input into non-orthogonal basis state $\rho_{i_{0}, k} = u_{k} | 0 \rangle \langle 0| + \frac{1-u_{k}}{2} (| 0 \rangle + | 1 \rangle) (\langle 0 | + \langle 1 |)$. We denote these different input encodings as mixed, pure, phase and non-orthogonal. We emphasize that for the last three encodings the universality of the associated dissipative quantum system using these encodings has not been proven. 

To investigate the impact of input encodings on the computational capability of quantum systems, the Hamiltonian parameters for all quantum systems simulated here are sampled from the same uniform distribution, and the resulting Hamiltonians are chosen if the associated CPTP map that implements the specified input-dependent density operator $\rho_{i_{0}, k}$ is convergent. We again test the convergence property numerically by checking if $50$ randomly generated initial density operators converge to the same density operator within $500$ timesteps. The number of system qubits and the number of computational nodes for all input encodings are the same. For each input encoding, NMSEs of $100$ convergent quantum systems are averaged for analysis. Fig.~\ref{figure:input_encoding} shows that for all computational tasks, the mixed state encoding performs better than other encodings. However, the average NMSE for different input encodings for all computational tasks are of the same order of magnitude. Moreover, as the number of system qubits increases, the errors of different input encodings decrease at roughly the same rate. This comparison indicates that the effect of different input encodings on the learning performance does not appear significant.

\begin{figure}[htb!]
\centering
\includegraphics[trim={40mm 0mm 35mm 5mm}, clip, scale=0.43]{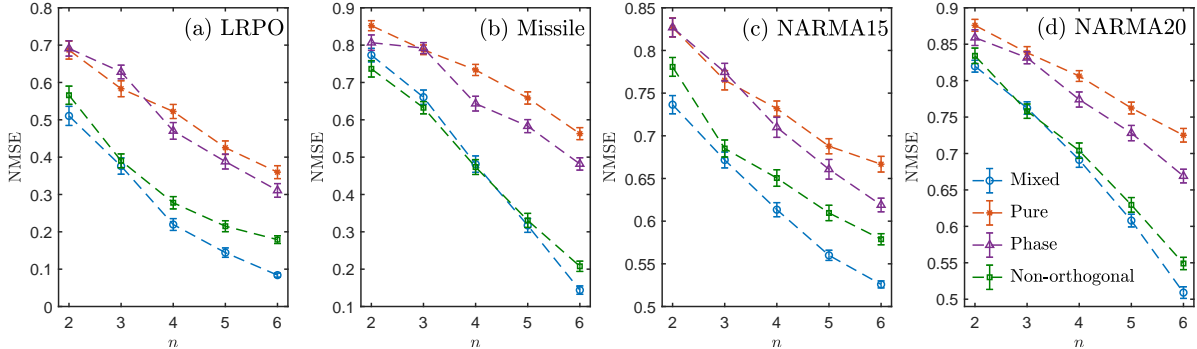}
\caption{Average NMSE for different input encodings, for approximating the (a) LRPO, (b) Missile (c) NARMA15 and (d) NARMA20 tasks. Error bars represent the standard error}
\label{figure:input_encoding}
\end{figure}

\subsection{Further comparison with ESNs} \label{subsec:further_comparison}
Our numerical results so far and the results shown in \cite{fujii2017harnessing} both suggest that dissipative quantum systems with a small number of qubits achieve comparable performance to classical learning schemes with a large number of computational nodes. However, these comparisons may appear to be skewed favorably toward quantum dynamical systems, since it does not address their exponential state-space size. One can, for example, also increase the state-space size of ESNs and the number of computational nodes of SA, such that the state-space size and the number of computational nodes are similar for both models. Here we present a further comparison between the SA model and ESNs, and provide insights into the possible advantage the SA model might offer over its classical counterpart. 

We focus on 4-qubit SA with a state-space size of $256$. Setting $R=6$ in Eq.~\eqref{eq:mv-out}, the number of computational nodes for SA is $210$. We compare this 4-qubit SA model's average task performance with the average E256 task performance in approximating the LRPO, Missile, NARMA15, NARMA20, NARMA30 and NARMA40 tasks. Here, the number of computational nodes for E256 is 257 and the average NMSE of 100 convergent E256s is reported. As shown in Table~\ref{table:q4-comparison}, for the Missile and all the NARMA tasks, the average NMSEs for both models are of the same order of magnitude, while E256 outperforms SA in the LRPO task. This comparison suggests that when the state-space size and the number of computational nodes for both models are similar, ESNs can outperform the SA model.

\begin{table}
\caption{Average $4$-qubit SA and E256 NMSE for the LRPO, Missile, NARMA15, NARMA20, NARMA30 and NARMA40 tasks. Results are rounded to two significant figures. The notation ($\pm$ se) denotes the standard error}
\label{table:q4-comparison}
\begin{tabular}{ccc}
\hline\noalign{\smallskip}
Task     & SA NMSE ($\pm$ se)    &  E256 NMSE ($\pm$ se) \\
\noalign{\smallskip}\hline\noalign{\smallskip}
LRPO     &  $0.20 \pm  1.5 \times 10^{-2}$   &  $0.019  \pm 7.7 \times 10^{-4}$ \\
Missile  &  $0.48 \pm  2.2 \times 10^{-2}$   &  $0.49   \pm 3.3 \times 10^{-3}$ \\
NARMA15  &  $0.61 \pm  8.0 \times 10^{-3}$   &  $0.32   \pm 1.6 \times 10^{-4}$ \\
NARMA20  &  $0.68 \pm  1.0 \times 10^{-2}$   &  $0.67   \pm 3.2 \times 10^{-4}$ \\
NARMA30  &  $0.67 \pm  7.1 \times 10^{-3}$   &  $0.67   \pm 4.0 \times 10^{-4}$ \\
NARMA40  &  $0.64 \pm  5.3 \times 10^{-3}$   &  $0.66   \pm 5.9 \times 10^{-4}$ \\
\noalign{\smallskip}\hline
\end{tabular}
\end{table}

We further investigate under what conditions SA might offer a computational advantage. We observe that while the number of computational nodes is kept constant, increasing the state-space size of SA induces a considerable computational improvement. To demonstrate this, $4$-, $5$- and $6$-qubit SA samples are simulated to perform all computational tasks mentioned above. For each $n$-qubit SA, we vary its output degree $R$ such that its number of computational nodes ranges from 5 to 252. The chosen degrees for $4$-qubit SA are $R_{4} = \{1, \ldots, 6\}$, for $5$-qubit are $R_{5} = \{1, \ldots, 5\}$, and for $6$-qubit SA are $R_{6}=\{1, \ldots, 4\}$. For each $n$-qubit SA, the Hamiltonians are the same for all its chosen output degrees, and the task errors of 100 convergent SA samples are averaged for analysis.

For comparison, we simulate 100 convergent ESNs with reservoir size $256$ to perform the same tasks. For $n$-qubit SA, let $\mathcal{N}_{n}$ $(n=4, 5, 6)$ denote the numbers of computational nodes corresponding to its output degrees $R_{n}$. The number of computational nodes $\mathcal{C}$ for E256 is set to be elements in the set $\mathcal{N}_{4} \cup \mathcal{N}_{5} \cup \mathcal{N}_{6}$. We first optimize $257$ output weights for E256 via standard least squares during the training phase. When $\mathcal{C}<257$ for E256, we select $\mathcal{C}-1$ computational nodes (excluding the tunable constant computational node) with the largest absolute values and their corresponding state elements. These $\mathcal{C}-1$ state elements are used to re-optimize $\mathcal{C}$ computational nodes (including the tunable constant computational node) via standard least squares. During the evaluation phase, $256$ state elements evolve; only $\mathcal{C}-1$ state elements and $\mathcal{C}$ output weights are used to compute the E256 output.

\begin{figure}[ht!]
\begin{center}
\includegraphics[trim={30mm 15mm 35mm 10mm}, clip, scale=0.43]{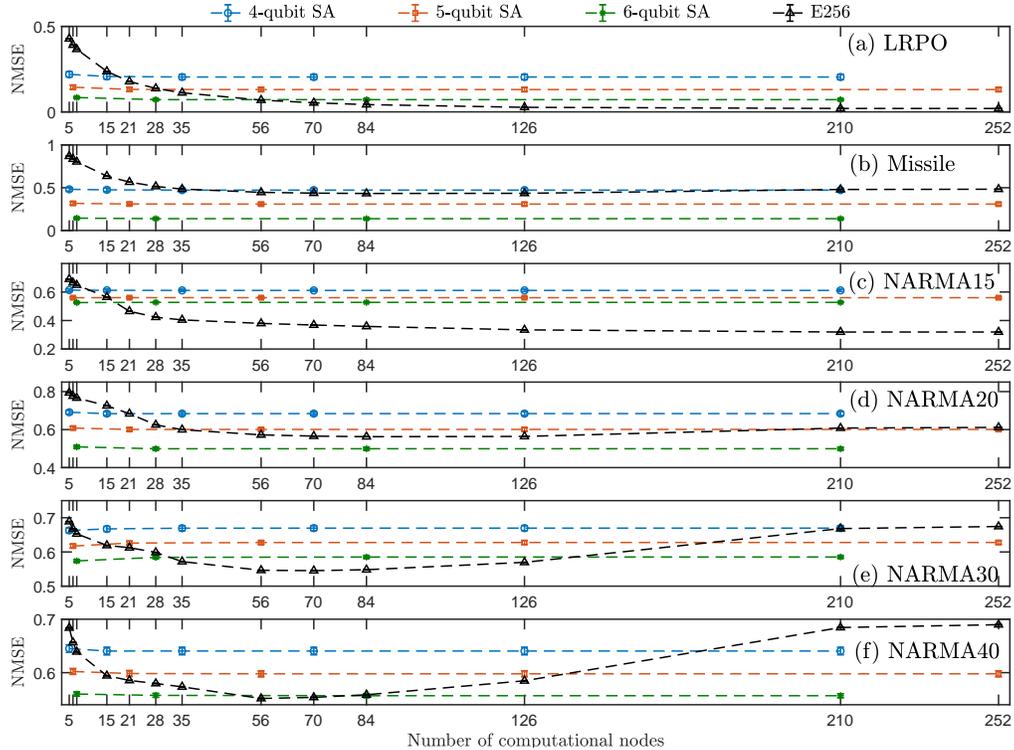}
\caption{Average SA NMSE as the state-space size and the number of computational nodes vary for all computational tasks. The average NMSE for E256 with the same number of computational nodes is plotted for comparison. The data symbols obscure the error bars, which represent the standard error}
\label{figure:SA_SS}
\end{center}
\end{figure}

\begin{figure}[ht!]
\begin{center}
\includegraphics[trim={25mm 15mm 30mm 10mm}, clip, scale=0.47]{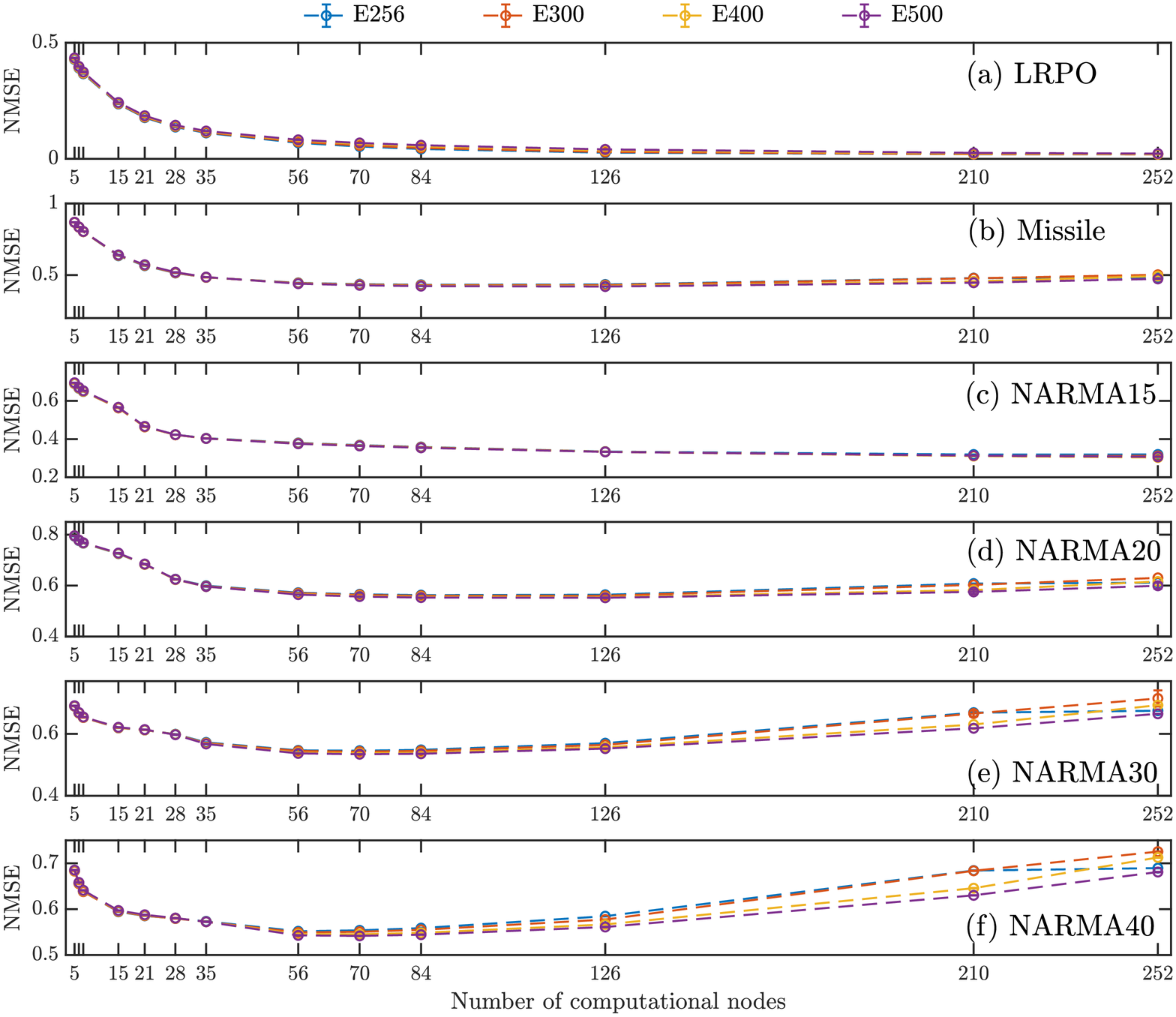}
\caption{Average ESNs NMSE as the state-space size and the number of computational nodes vary for all computational tasks. The data symbols obscure the error bars, which represent the standard error}
\label{figure:ESN_SS}
\end{center}
\end{figure}

Fig.~\ref{figure:SA_SS} plots the 4-, 5-, and 6-qubit SA average NMSE as the number of computational nodes increases for all computational tasks. For comparison, the average E256 NMSE is also plotted. Two important observations are that increasing the number of computational nodes does not necessarily improve SA task performance, while increasing the state-space size induces a noticeable improvement. For example, for the NARMA20 task and 210 computational nodes, the average NMSE for $4$-qubit SA is 0.68 while the average NMSE for $6$-qubit SA is 0.48. When comparing to E256, we observe that for most tasks, despite 4-qubit SA might not perform better than E256, subsequent increases in the state-space size allow the SA model to outperform E256, without extensively increasing its number of computational nodes. 

Contrary to the above observations for the SA model, increasing the reservoir size of ESNs while keeping the number of computational nodes fixed does not induce a significant computational improvement. To numerically demonstrate this, the reservoir size of ESNs is further increased to \{300, 400, 500\}. For each reservoir size, the number of computational nodes is set to be the same as that of E256. These computational nodes are chosen and optimized by the same method described above for E256. We average the task errors of 100 convergent ESNs for each reservoir size.  As shown in Fig.~\ref{figure:ESN_SS}, noticeable performance improvements for ESNs are only observed as the number of computational nodes increases, but not as the reservoir size varies. Another observation is that for the NARMA30 and NARMA40 tasks, the error increases as the number of computational nodes for ESNs increases. This could be due to overfitting, a condition occurs when too many adjustable parameters are trained on limited training data \cite{friedman2001elements}. On the other hand, this observation is less significant for the SA model. It would be interesting to conduct further investigation into this behavior in future work.

The above observations have several implications. To improve the computational capability of the SA model, one can take advantage of the exponentially increasing state-space size of the Hilbert space while only optimizing a polynomial number of computational nodes. On the contrary, to improve the computational capability of ESNs, one needs to increase the number of computational nodes, which is bounded by the reservoir size. Therefore, enhancing emulation performance of ESNs inevitably requires the state-space size to be increased. In the situation where the state-space size increases beyond what classical computers can simulate in a reasonable amount of time and with reasonable resources (such as memory), the computational capability of ESNs saturates, whereas the computational capability of the SA model could be further improved by increasing the number of qubits in a linear fashion. In this regime, the SA model could provide a potential computational advantage over its classical counterpart. To further verify the feasibility of this hypothesis, the learning capability of the SA model would need to be evaluated for a larger number of qubits on a physical quantum system. A possible implementation of this experiment is on NMR ensembles, as suggested in \cite{fujii2017harnessing}. However, motivated by the availability of NISQ machines, a quantum circuit implementation of the SA model, using the schemes proposed in \cite{BvH09,Gross18}, would be more attractive. This is another topic of further research continuing from this work.

\section{Discussion}
\label{sec:discussion}
We discuss the feasibility of realizing the proposed universal scheme in Section \ref{sec:universal_class} on the current most scalable NISQ quantum computers, such as quantum computers based on superconducting circuits or ion traps. We consider those  that implement the quantum circuit model.  Simulating the unitary interaction given by the Ising Hamiltonian $H_{K}$ on a quantum circuit requires decomposition of the evolution using the Trotter-Suzuki product formula \cite{trotter1959product,suzuki1971relationship}. Such a decomposition may require the sequential application of a large number of gates on NISQ machines, limiting the implementability of this family on current NISQ machines due to severe decoherence. However, it  may be possible to engineer alternative families based on simpler unitary interactions between the subsystem and ancilla qubits  (not of the Ising type), using only a short sequence of single-qubit and two-qubit gates, such that the associated CPTP maps possess the convergence and fading memory properties. A general framework for constructing such unitary interactions is the subject of on-going and future work continuing from this paper. 

To realize the dissipative dynamics for a subsystem, we can construct a quantum circuit as follows. At each timestep $k$, the ancilla qubit is prepared as the input-dependent mixed state $\rho_{i_0, k}^{K}$. After the unitary interaction with the subsystem qubits, the partial trace over the ancilla qubit can be performed by a projective measurement on the ancilla qubit and discarding the measurement outcome. At the next time step $k+1$, the ancilla qubit is reset and prepared as $\rho_{i_0, k+1}^{K}$. To estimate the expectations $\langle Z^{(j)} \rangle_{k}$, we perform Monte Carlo estimation by running the circuit multiple times and measuring $Z^{(j)}$ at time $k$,  the average of measured results over these runs estimates $\langle Z^{(j)} \rangle_{k}$. If multiple NISQ machines can be run in parallel at the same time, the expectations $\langle Z^{(j)} \rangle_{k}$ can be estimated in real time. In this setting, the number of qubits required to implement a dissipative subsystem for temporal learning is $n_{K}+1$.

Some existing NISQ machines based on superconducting circuits are not capable of preparing mixed states or resetting qubits for reuse after measurement. To address the first limitation, we can approximate the ancilla input-dependent mixed state $\rho_{i_0, k}^{K}$ by Monte Carlo sampling. That is, we construct $M > 0$ quantum circuits as above, but for each circuit and at each timestep $k$, we prepare the ancilla qubit in $|0 \rangle$ with probability $u_k$ or in $|1 \rangle$ with probability $1-u_k$. We again remark that these $M$ quantum circuits can be run in parallel, and therefore computations could be performed in real time if multiple circuits can be run at the same  time. Not being able to reuse a qubit after a measurement means that each point in the input sequence must be encoded in a distinct qubit. This implies that the length of sequences that can be considered will be limited by the number of qubits available. Some of the qubits available will need to be assigned as the system qubits while all the other qubits as data carrying ancilla qubits. For instance, on a 20-qubit machine, one can use say 4 qubits as the system qubits and the remaining 16 qubits for carrying the input sequence. In this case the total input sequence length that can be processed for washout, learning and evaluation is only of length 16. Nevertheless, current high-performance quantum circuit simulators, such as the IBM Qiskit simulator (https://qiskit.org/) \cite{Qiskit}, are capable of simulating qubit reset and realistic hardware noise. We also anticipate that the qubit reset functionality on NISQ machines would be available in the near future, opening avenue for proof-of-principle experiments of the proposed scheme for input sequences of arbitrary length.

\section{Conclusion and outlook}
We have developed a general theory for learning arbitrary I/O maps with fading memory using dissipative quantum systems. The attractiveness of the theory studied here is that it allows a dissipative quantum system (that meets certain requirements but is otherwise arbitrary) to be combined with a classical processor to learn I/O maps from sample I/O sequences. We apply the theory to demonstrate a universal class of dissipative quantum systems that can approximate arbitrary I/O maps with fading memory. 

Numerical experiments indicate that even with only a small number of qubits and a simple linear output, this class can achieve comparable performance, in terms of the average normalized mean-squared error, to classical learning schemes such as ESNs and the Volterra series with a large number of tunable parameters. However, when the state-space sizes of the quantum subsystem and classical learning schemes are the same, and the number of computational nodes equals the number of nodes in the ESN plus one (for the constant term) and a similar number of the QRC, the quantum system does not demonstrate any computational advantage. Moreover, the numerical results for a small number of qubits indicate that increasing the dimension of the Hilbert space of the quantum system while fixing the number of computational nodes can still result in improved prediction performance on a number of benchmarking tasks, whereas increasing the state space of ESNs while fixing the computational nodes does not lead to any noticeable improvement. This strongly suggests that the possibly very large Hilbert space of the quantum subsystem presents a potential resource that can be exploited in this approach. That is, for state-space dimensions beyond what can be simulated on a conventional digital computer. It remains to be investigated if this resource can indeed lead to a provable performance advantage over conventional classical learning approaches, and the circumstances where this will be the case. 

%\begin{acknowledgements}
%\end{acknowledgements}

\section{Appendix}

\subsection{The convergence property} \label{app:covergence}
Recall from the main text that for a compact subset $D \subseteq \mathbb{R}$ and $L>0$, $K_L(D)$ denotes the set of all real sequences $u=\{u_k\}_{k \in \mathbb{Z}}$ taking values in $D \cap  [-L,L]$. Let $K_{L}^{-}(D)$ and $K_{L}^{+}(D)$ be subsets of input sequences in $K_{L}(D)$  whose indices are restricted to $\mathbb{Z}^{-}=\{\ldots, -2, -1, 0\}$ and $\mathbb{Z}^{+} = \{1, 2, \ldots \}$, respectively. In the following, we write $T$ for both input-independent and input-dependent CPTP maps. As in the main text, we write $T(u_k)$ for a CPTP map that is determined by an input $u_k$, and $\| \cdot \|_{p}$ for any Schatten $p$-norm for $p \in [1, \infty)$. All dissipative quantum systems considered here are finite-dimensional. We now state the definition of a convergent CPTP map with respect to $K_{L}(D)$.

\begin{definition}[Convergence]
An input-dependent CPTP map $T$ is convergent with respect to $K_{L}(D)$ if there exists a sequence $\delta = \{\delta_k\}_{k > 0}$ with $\mathop{\lim}_{k \rightarrow \infty} \delta_{k} = 0$, such that for all $u = \{u_{k}\}_{k \in \mathbb{Z}^{+}} \in K_{L}^{+}(D)$ and any two density operators $\rho_{j, k}$ $(j = 1, 2)$ satisfying $\rho_{j, k} = T(u_{k})\rho_{j, k-1}$, it holds that $\left\| \rho_{1, k} - \rho_{2, k} \right\|_{2} \leq \delta_{k}$. We call a dissipative quantum system whose dynamics is governed by a convergent CPTP map a convergent system.
\end{definition}

The convergence property can be viewed as an extension of the mixing property for a noisy quantum channel described by an input-independent CPTP map \cite{richter1996ergodicity}. 

\begin{definition}[Mixing]
\label{defn:mixing-CPTP}
A $n$-qubit dissipative quantum system described by a CPTP map $T$ is mixing if for all $\rho_{0} \in \mathcal{D}(\mathbb{C}^{2^n})$, if there exists a unique density operator $\rho_{*}$ such that,
$$\mathop{\lim}_{k \rightarrow \infty} \left\| \left( \prod_{j=1}^{k} T(\rho_{0}) \right) - \rho_{*} \right\|_{2} = 0.$$
\end{definition}

We will see later that if an input-dependent CPTP map $T(u_{k})$ satisfies the sufficient condition in Theorem~\ref{theorem:convergence}, then $T(u_k)$ is mixing for each $u_k \in D \cap [-L, L]$.

\begin{theorem}[Convergence property]
\label{theorem:convergence}
A $n$-qubit dissipative quantum system governed by an input-dependent CPTP map $T$ is convergent with respect to $K_{L}(D)$ if, for all $u_{k} \in D \cap [-L, L]$, $T(u_k)$ on the hyperplane $H_{0}(2^n)$ of $2^n \times 2^n$ traceless Hermitian operators satisfies $\|T(u_k)\rvert_{H_{0}(2^n)} \|_{2-2} \coloneqq \sup_{A \in H_{0}(2^n), A \neq 0} \frac{\| T(u_k) A\|_{2}}{\|A\|_{2}} \leq 1 - \epsilon$ for some $0 < \epsilon \leq 1$. Moreover, any pair of initial density operators converge uniformly to one another under $T$.
\end{theorem}

\begin{proof}
Let $\rho_{1, 0}$ and $\rho_{2, 0}$ be two arbitrary initial density operators, $\rho_{1, 0} - \rho_{2, 0}$ is a traceless Hermitian operator. We have,
\begin{equation*}
\begin{split}
\left\| \rho_{1, k} - \rho_{2, k} \right\|_{2} & =  \left\| \left( \overleftarrow{\prod}_{j=1}^{k} T(u_{j}) \right) (\rho_{1, 0} - \rho_{2, 0}) \right\|_{2} \\
& = \left\| \left( \overleftarrow{\prod}_{j=1}^{k} T(u_{j}) \rvert_{H_{0}(2^n)} \right) (\rho_{1, 0} - \rho_{2, 0}) \right\|_{2} \\
& \leq  \overleftarrow{\prod}_{j=1}^{k}  \left\|T(u_{j}) \rvert_{H_{0}(2^n)} \right\|_{2-2} \left\| \rho_{1, 0} - \rho_{2, 0} \right\|_{2} \\
& \leq \overleftarrow{\prod}_{j=1}^{k} (1-\epsilon) \left\| \rho_{1, 0} - \rho_{2, 0} \right\|_{2} \\
& \leq \overleftarrow{\prod}_{j=1}^{k} (1 - \epsilon) (\left\| \rho_{1, 0} \right\|_{2} + \left\|\rho_{2, 0} \right\|_{2}) \\
& \leq  2(1-\epsilon)^{k}, 
\end{split}
\end{equation*}
where the last inequality follows from the fact that for all $\rho \in \mathcal{D}(\mathbb{C}^{2^n})$, $\left\|\rho \right\|_{2} \leq 1$.
\end{proof}

We remark that for a $n$-qubit dissipative quantum system that satisfies the condition in Theorem~\ref{theorem:convergence}, any initial density operator $\rho_{0}$ reaches the state $\lim_{k \rightarrow \infty} \left( \overleftarrow{\prod}_{j=1}^{k}T(u_j) \right) \left(\frac{I}{2^{n}}\right)$. To see this, let 
$$
\rho_{0} = \frac{I}{2^n} + \sum_{\substack{j_1, j_2,\ldots,j_n = \{0, 1, 2, 3\} \\ j_1j_2\ldots j_n \neq 0}} \alpha_{j_1j_2\ldots j_n}\bigotimes_{i=1}^{n} \sigma^{(i)}_{j_i},
$$
where $\sigma^{(i)}_{j_{i}}$ denotes, for qubit $i$, the identity operator $I$ if $j_{i}=0$, the Pauli X operator if $j_{i} = 1$, the Pauli Y operator if $j_{i}=2$ and the Pauli Z operator if $j_{i}=3$. Since $\bigotimes_{i=1}^{n} \sigma^{(i)}_{j_i}$ for $j_1j_2\ldots j_n \neq 0$ are all traceless Hermitian operators, therefore as $k \rightarrow \infty$, 
\begin{equation*}
\begin{split}
\left\| \rho_{k} - \left( \overleftarrow{\prod}_{j=1}^{k} T(u_j) \right) \left( \frac{I}{2^n} \right) \right\|_{2} \rightarrow 0.
\end{split}
\end{equation*}

\subsection{The universality property} \label{app:universality}
We now show the universality property of convergent dissipative quantum systems. Let $\mathbb{R}^{\mathbb{Z}}$ be the set of all real-valued infinite sequences. Consider a $n$-qubit convergent dissipative quantum system described by Eqs.~\eqref{eq:rho-dynamics} and \eqref{eq:y-dynamics} in the main text, whose dynamics and output are defined by a CPTP map $T$ and a functional $h: \mathcal{D}(\mathbb{C}^{2^n}) \rightarrow \mathbb{R}$, respectively. We associate this quantum system with an induced filter $M^{T}_{h}: K_{L}(D) \rightarrow \mathbb{R}^{\mathbb{Z}}$, such that for any initial condition $\rho_{-\infty} \in \mathcal{D}(\mathbb{C}^{2^n})$, when evaluated at time $t = k \tau$,
$$
M_{h}^{T}(u)_{k} =  h\left( \left( \overrightarrow{\prod}_{j=0}^{\infty} T(u_{k-j}) \right) \rho_{-\infty} \right),
$$
where $\overrightarrow{\prod}_{j=0}^{\infty} T(u_{k-j}) = \mathop{\lim}_{N \rightarrow \infty} \overleftarrow{\prod}_{j=0}^{N} T(u_{k+(j-N)})=\lim_{N \rightarrow \infty} T(u_{k}) T(u_{k-1}) \cdots T(u_{k-N})$, and the limit is a pointwise limit. Lemma~\ref{lemma:limit} states that this limit is well-defined.

\begin{lemma}
\label{lemma:limit}
The filter $M_{h}^{T}: K_{L}(D) \rightarrow \mathbb{R}^{\mathbb{Z}}$ is well-defined. In particular, the limit $\lim_{N \rightarrow \infty} T(u_{k})T(u_{k-1}) \cdots T(u_{k-N})\rho_{-N}$ exists and is independent of $\rho_{-N}$.
\end{lemma}
\begin{proof}
The set $\mathcal{D}(\mathbb{C}^{2^n})$ equipped with the distance function induced by the norm $\| \cdot \|_{2}$ is a complete metric space. Therefore, every Cauchy sequence converges to a point in $\mathcal{D}(\mathbb{C}^{2^n})$ \cite{rudin1964principles}. It remains to show that $S_{n} = T(u_{k})T(u_{k-1}) \cdots T(u_{k-n}) \rho_{-n}$ is a Cauchy sequence. By hypothesis, for all $u_{k} \in D \cap [-L, L]$, $\| T(u_{k})\rvert_{H_{0}(2^{n})}\|_{2} \leq 1 - \epsilon$ for some $0 < \epsilon \leq 1$. For any $\epsilon' > 0$, let $N>0$ such that $(1-\epsilon)^{N} < \frac{\epsilon'}{2}$. Then for all $n, m > N$, suppose that $n \leq m$,
\begin{equation*}
\begin{split}
\|S_{n} - S_{m}\|_{2} & = \left\|T(u_{k})T(u_{k-1})\cdots T(u_{k-n}) \left(\rho_{-n} - T(u_{k-n-1})\cdots T(u_{k-m}) \rho_{-m}\right) \right\|_{2} \\
& \leq (1-\epsilon)^{n+1} \left(\|\rho_{-n}\|_{2} +\|  \left( T(u_{k-n-1})\cdots T(u_{k-m}) \right) \rho_{-m}\|_{2} \right) \\
& \leq 2(1-\epsilon)^{N} < \epsilon'
\end{split}
\end{equation*}
\end{proof}

This filter is causal since given $u, v \in K_{L}(D)$ satisfying $u_{\tau} = v_{\tau}$ for $\tau \leq k$, $M_{h}^{T}(u)_{k} = M_{h}^{T}(v)_{k}$. For any $\tau \in \mathbb{Z}$, let $M_{\tau}$ be the shift operator defined by $M_{\tau}(u)_{k} = u_{k-\tau}$. A filter is said to be time-invariant if it commutes with $M_{\tau}$. It is straightforward to show that $M_{h}^{T}$ is time-invariant. 

For a time-invariant and causal filter, there is a corresponding functional $F_{h}^{T}: K_{L}^{-}(D) \rightarrow \mathbb{R}$ defined as $F_{h}^{T}(u) = M_{h}^{T}(u)_0$ (see \cite{BC85}). The corresponding filter can be recovered via $M_{h}^{T}(u)_{k} = F_{h}^{T}(P \circ M_{-k}(u)) $, where $P$ truncates $u$ up to $0$, that is $P(u) = u\rvert_0$. We say a filter $M_{h}^{T}$ has the fading memory property if and only if $F_{h}^{T}$ is continuous with respect to a weighted norm defined as follows.

\begin{definition}[Weighted norm]
\label{defn:weighted_norm}
For a null sequence $w=\{w_{k}\}_{k \geq 0}$, that is $w: \{0\} \cup \mathbb{Z}^{+}  \rightarrow (0, 1]$ is decreasing and $\mathop{\lim}_{k \rightarrow \infty} w_{k} = 0$, define a weighted norm $\|\cdot\|_{w}$ on $K^{-}_{L}(D)$ as $\|u\|_{w} \coloneqq \mathop{\sup_{k \in \mathbb{Z}^{-}}} \left|u_{k} \right| w_{-k}$.
\end{definition}

\begin{definition}[Fading memory]
A time-invariant causal filter $M: K_{L}(D) \rightarrow \mathbb{R}^{\mathbb{Z}}$ has the fading memory property with respect to a null sequence $w$ if and only if its corresponding functional $F: K_{L}^{-}(D) \rightarrow \mathbb{R}$ is continuous with respect to the weighted norm $\|\cdot\|_{w}$.
\end{definition}

To emphasize that the fading memory property is defined with respect to a null sequence $w$, we will say that $M$ is a \textit{$w$-fading memory filter} and the corresponding functional $F$ is a \textit{$w$-fading memory functional}. We state the following compactness result \cite[Lemma 2]{grigoryeva2018universal} and the Stone-Weierstrass theorem \cite[Theorem 7.3.1]{Dieudonne13}.

\begin{lemma}[Compactness]
\label{lemma:compactness}
For any null sequence $w$, $K_{L}^{-}(D)$ is compact with the weighted norm $\| \cdot \|_{w}$.
\end{lemma}
We write $(K_{L}^{-}(D), \|\cdot\|_{w})$ to denote the space $K_{L}^{-}(D)$ equipped with the weighted norm $\| \cdot \|_{w}$.

\begin{theorem}[Stone-Weierstrass]
\label{theorem:stone}
Let $E$ be a compact metric space and $C(E)$ be the set of real-valued continuous functions defined on $E$. If a subalgebra $A$ of $C(E)$ contains the constant functions and separates points of $E$, then $A$ is dense in $C(E)$.
\end{theorem}

Let $C(K_{L}^{-}(D), \| \cdot \|_{w})$ be the set of continuous functionals $F: (K_{L}^{-}(D),  \| \cdot \|_{w}) \rightarrow \mathbb{R}$. The following theorem is a result of the compactness of $(K_{L}^{-}(D), \| \cdot \|_{w})$ (Lemma \ref{lemma:compactness}) and the Stone-Weierstrass Theorem (Theorem \ref{theorem:stone}).

\begin{theorem}
\label{theorem:universal_dissipative}
Let $w$ be a null sequence. For convergent CPTP maps $T$, let $\mathcal{M}_{w} = \{M_{h}^{T}  \mid h : \mathcal{D}(\mathbb{C}^{2^{n}}) \rightarrow \mathbb{R}\}$ be a set of $w$-fading memory filters. Let $\mathcal{F}_{w}$ be the family of corresponding  $w$-fading memory functionals defined on $K_{L}^{-}(D)$. If $\mathcal{F}_{w}$ forms a polynomial algebra of $C(K_{L}^{-}(D), \| \cdot \|_{w})$, contains the constant functionals and separates points of $K_{L}^{-}(D)$, then $\mathcal{F}_{w}$ is dense in $C(K_{L}^{-}(D), \| \cdot \|_{w})$. 
That is for any $w$-fading memory filter $M_{*}$ and any $\epsilon > 0$, there exists $M_{h}^{T} \in \mathcal{M}_{w}$ such that for all $u \in K_{L}(D)$,  $\|M_{*}(u) - M_{h}^{T}(u)\|_{\infty} = \sup_{k \in \mathbb{Z}}|M_{*}(u)_{k} - M_{h}^{T}(u)_{k}| < \epsilon$.
\end{theorem}

\begin{proof}
$\mathcal{F}_{w}$ is dense follows from Lemma \ref{lemma:compactness} and Theorem \ref{theorem:stone}. To prove the second part of the theorem, since $\mathcal{F}_{w}$ is dense in $C(K_{L}^{-}(D), \| \cdot \|_{w})$, for any $w$-fading memory functional $F_{*}$ and any $\epsilon>0$ , there exists $F_{h}^{T} \in \mathcal{F}_{w}$ such that for all $u_{-} \in K_{L}^{-}(D)$,  
$
| F_{*}(u_{-}) - F_{h}^{T}(u_{-})| < \epsilon.
$
For $u \in K_{L}(D)$, notice that $P \circ M_{-k}(u) \in K_{L}^{-}(D)$ for all $k\in \mathbb{Z}$, hence
\begin{equation*}
\begin{split}
\left| F_{*}(P \circ M_{-k}(u)) - F_{h}^{T}(P \circ M_{-k}(u)) \right| = \left| M_{*}(u)_{k} - M_{h}^{T}(u)_{k} \right| < \epsilon.
\end{split}
\end{equation*}
Since this is true for all $k \in \mathbb{Z}$, therefore for all $u \in K_{L}(D)$, $\|M_{*}(u) - M_{h}^{T}(u) \|_{\infty} < \epsilon$.
\end{proof}

\subsection{Fading memory property and polynomial algebra} \label{app:FMP_algebra}
Before we prove the universality of the family of dissipative quantum systems introduced in Sec.~\ref{sec:universal_class} in the main text, we first show two important observations regarding to the multivariate polynomial output in Eq.~\eqref{eq:mv-out}. 

We specify $h$ to be the multivariate polynomial as in the right hand side of Eq.~\eqref{eq:mv-out} in the main text. For ease of notation, we drop the subscript $h$ in $F_{h}^{T}$ and $M_{h}^{T}$. Let $\mathcal{F} = \{F^{T}\}$ be the set of functionals induced from dissipative quantum systems given by Eqs.~\eqref{eq:rho-dynamics} and \eqref{eq:mv-out} in the main text. We will show in Lemma~\ref{lemma:FMP} that the convergence and continuity of $T$ are sufficient to guarantee the fading memory property of $F^{T}$, and in Lemma~\ref{lemma:algebra} that $\mathcal{F}$ forms a polynomial algebra, made of fading memory functionals. In the following, let $\mathcal{L}(\mathbb{C}^{2^n})$ be the set of linear operators on $\mathbb{C}^{2^n}$, and for a CPTP map $T$, for all $u_k \in D \cap [-L, L]$, define $\|T(u_k)\|_{2-2} \coloneqq \sup_{A \in \mathcal{L}(\mathbb{C}^{2^n}), \|A\|_{2}=1} \|T(u_k) A\|_{2}$.

\begin{lemma}[Fading memory property]
\label{lemma:FMP}
Consider a $n$-qubit dissipative quantum system with dynamics Eq.~\eqref{eq:rho-dynamics} and output Eq.~\eqref{eq:mv-out}. Suppose that for all $u_k \in D \cap [-L, L]$, the CPTP map $T(u_k)$ satisfies the condition in Theorem \ref{theorem:convergence}, so that it is convergent. Moreover, for any $\epsilon > 0$, there exists $\delta_{T}(\epsilon) > 0$ such that $\|T(x) - T(y)\|_{2-2} < \epsilon$ whenever $|x - y| < \delta_{T}(\epsilon)$ for $x, y \in D \cap [-L, L]$. Then for any null sequence $w$, the induced filter $M^{T}$ and the corresponding functional $F^{T}$ are $w$-fading memory.
\end{lemma}                                                                            
\begin{proof}
We first state the boundedness of CPTP maps \cite[Theorem 2.1]{perez2006contractivity}.
\begin{lemma}
\label{lemma:bounded-CPTP}
For a CPTP map $T : \mathcal{L}(\mathbb{C}^{2^n}) \rightarrow \mathcal{L}(\mathbb{C}^{2^n})$, we have $\left\| T \right\|_{2-2} \leq \sqrt{2^n}$.
\end{lemma}
Moreover, recall that ${\rm Tr}(\cdot)$ is continuous, that is for any $\epsilon>0$, there exists $\delta_{\rm Tr}(\epsilon) > 0$ such that $|{\rm Tr}(A-B)|<\epsilon$ whenever $\|A - B \|_{2} <\delta_{\rm Tr}(\epsilon)$ for any complex matrices $A, B$. Note that here $\| \cdot \|_{2}$ denotes the Schatten $2$-norm or the Hilbert-Schmidt norm.

Let $w$ be an arbitrary null sequence. We will show the linear terms $L(u)$ in the functional $F^{T}$ are continuous with respect to $\left\| \cdot \right\|_{w}$, and the continuity property of $F^T$ follows from the fact that finite sums and products of continuous elements are also continuous.

For any $u, v \in K_{L}^{-}(D)$,
\begin{equation*}
\begin{split}
 \left| L(u) - L(v) \right| = \left| {\rm Tr} \left( Z^{(i_{1})} \left( \left(\overrightarrow{\prod}_{k=0}^{\infty}T(u_{-k}) \right) \rho_{-\infty} -  \left(\overrightarrow{\prod}_{k=0}^{\infty}T(v_{-k}) \right) \rho_{-\infty} \right) \right) \right|.
\end{split}
\end{equation*}
Denote $\rho_{u} = \left( \overrightarrow{\prod}_{k=N}^{\infty}T(u_{-k}) \right)\rho_{-\infty}$ and $\rho_{v} = \left(\overrightarrow{\prod}_{k=N}^{\infty}T(v_{-k}) \right) \rho_{-\infty}$ for some $0 < N < \infty$,
\begin{equation}
\label{eq:norm_of_linear_term}
\begin{split}
& \left\| Z^{(i_{1})} \left( \left( \overrightarrow{\prod}_{k=0}^{\infty} T(u_{-k}) \right) \rho_{-\infty} -  \left( \overrightarrow{\prod}_{k=0}^{\infty} T(v_{-k}) \right) \rho_{-\infty} \right)\right\|_{2} \\
& \leq \left\|Z^{(i_{1})} \right\|_{2} \left( \left\| \overrightarrow{\prod}_{k=0}^{N-1} T(u_{-k})  - \overrightarrow{\prod}_{k=0}^{N-1}T(v_{-k}) \right\|_{2-2} \left\| \rho_{u} \right\|_{2}  + \left\| \left( \overrightarrow{\prod}_{k=0}^{N-1} T(v_{-k}) \right) (\rho_{u} - \rho_{v}) \right\|_{2} \right).
\end{split}
\end{equation}
Since $T$ satisfies conditions in Theorem \ref{theorem:convergence}, any two density operators converge uniformly to one another. Therefore, for any $\epsilon>0$, there exists $N(\epsilon) > 0$ such that for all $N' > N(\epsilon)$,
\begin{equation}
\label{eq:bound_1}
\left\| \left(\overrightarrow{\prod}_{k=0}^{N'-1}T(v_{-k}) \right) (\rho_{u} - \rho_{v}) \right\|_{2} < \frac{\delta_{\rm Tr}(\epsilon)}{2 \left\| Z^{(i_{1})} \right\|_{2}}.
\end{equation} 
Choose $N' = N(\epsilon) + 1$ and bound the first term inside the bracket on the right hand side of Eq.~\eqref{eq:norm_of_linear_term} by rewriting it as a telescopic sum:
\begin{equation}
\label{eq:FMP_first_term}
\begin{split}
& \left\| \overrightarrow{\prod}_{k=0}^{N(\epsilon)}T(u_{-k})  - \overrightarrow{\prod}_{k=0}^{N(\epsilon)}T(v_{-k})  \right\|_{2-2} \\
& = \left\| \sum_{l=0}^{N(\epsilon)} \left( T(v_{0}) \cdots T(v_{-(l- 1)}) T(u_{-l})  T(u_{-(l+1)}) \cdots T(u_{-N(\epsilon)}) \right. \right. \\
& \hskip 5em \left.  - T(v_{0}) \cdots T(v_{-(l-1)}) T(v_{-l}) T(u_{-(l+1)}) \cdots T(u_{-N(\epsilon)}) \right) \bigg\|_{2-2} \\
&\leq \sum_{l=0}^{N(\epsilon)} \left\| T(v_{0}) \cdots T(v_{-(l-1)}) \right\|_{2-2} \left\| T(u_{-l}) - T(v_{-l}) \right\|_{2-2} \left\|T(u_{-(l+1)}) \cdots T(u_{-N(\epsilon)}) \right\|_{2-2} \\
& \leq 2^{n} \sum_{l=0}^{N(\epsilon)} \left\| T(u_{-l}) - T(v_{-l}) \right\|_{2-2},
\end{split}
\end{equation}
where the last inequality follows from Lemma \ref{lemma:bounded-CPTP}. We claim that for any $\epsilon > 0$, if 
$$\|u - v \|_{w} = \mathop{\sup_{k \in \mathbb{Z}^{-}}} \left| u_{k} - v_{k} \right| w_{-k} < \delta_{T}\left( \frac{\delta_{\rm Tr} (\epsilon)}{2^{n+1} \left\| Z^{(i_1)} \right\|_{2} (N(\epsilon)+1)} \right) w_{N(\epsilon)}$$
then $\left| L(u) - L(v) \right| < \epsilon$. Indeed, since $w$ is decreasing, the above condition implies that
\begin{equation*}
\begin{split}
\mathop{\max_{0 \leq l \leq N(\epsilon)}} \left| u_{-l} - v_{-l} \right| w_{N(\epsilon)} < \delta_{T}\left( \frac{\delta_{\rm Tr} (\epsilon)}{2^{n+1} \left\| Z^{(i_1)} \right\|_{2} (N(\epsilon)+1)} \right)w_{N(\epsilon)}.
\end{split}
\end{equation*}
Since $w_{N(\epsilon)} > 0$, for all $0 \leq l \leq N(\epsilon)$,
$$
|u_{-l} - v_{-l}| < \delta_{T}\left( \frac{\delta_{\rm Tr} (\epsilon)}{2^{n+1} \left\| Z^{(i_1)} \right\|_{2} (N(\epsilon)+1)} \right).
$$
By continuity of $T$, we bound Eq.~\eqref{eq:FMP_first_term} by
\begin{equation}
\label{eq:bound_2}
\begin{split}
2^{n} \sum_{l=0}^{N(\epsilon)}\left\|T(u_{-l}) - T(v_{-l})  \right\|_{2-2} 
& < 2^{n} \sum_{l=0}^{N(\epsilon)} \frac{\delta_{\rm Tr} (\epsilon)}{2^{n+1} \left\| Z^{(i_1)} \right\|_{2} (N(\epsilon)+1)}  = \frac{\delta_{\rm Tr} (\epsilon)}{2 \left\| Z^{(i_1)} \right\|_{2}}.
\end{split}
\end{equation}
Since $\left\| \rho_{u} \right\|_{2} \leq 1$, Eqs.~\eqref{eq:norm_of_linear_term}, \eqref{eq:bound_1} and \eqref{eq:bound_2} give
\begin{equation*}
\begin{split}
&\left\|Z^{(i_1)} \right\|_{2} \left(\left\| \overrightarrow{\prod}_{k=0}^{N(\epsilon)}T(u_{-k}) -  \overrightarrow{\prod}_{k=0}^{N(\epsilon)}T(v_{-k}) \right\|_{2-2} \left\| \rho_{u} \right\|_{2}  + \left\| \left( \overrightarrow{\prod}_{k=0}^{N(\epsilon)}T(v_{-k}) \right) (\rho_{u} - \rho_{v}) \right\|_{2} \right) \\
& < \delta_{\rm Tr}(\epsilon).
\end{split}
\end{equation*}
The result now follows from the continuity of ${\rm Tr}(\cdot)$.
\end{proof}

\begin{lemma}[Polynomial algebra]
\label{lemma:algebra}
Let $\mathcal{F}=\{F^{T}\}$ be a family of functionals induced by dissipative quantum systems defined by Eqs.~\eqref{eq:rho-dynamics} and \eqref{eq:mv-out} in the main text. If for each member $F^{T} \in \mathcal{F}$, $T$ satisfies the conditions in Lemma~\ref{lemma:FMP}, then for any null sequence $w$, $\mathcal{F}$ forms a polynomial algebra consisting of $w$-fading memory functionals.
\end{lemma}

\begin{proof}
Consider two dissipative quantum systems described by Eqs.~\eqref{eq:rho-dynamics} and \eqref{eq:mv-out}, with $n_{1}$ and $n_{2}$ system qubits respectively. Let $\rho_{k}^{(m)} \in \mathcal{D}(\mathbb{C}^{2^{n_m}})$ be the state and $T^{(m)}$ be the CPTP map of the $m^{\rm th}$ system. Let $j_{1} = 1, \ldots, n_{1}$ and $j_{2} = 1, \ldots, n_{2}$ be the respective qubit indices for the two systems. For the observable $Z^{(j_{m})}$ of qubit $j_{m}$, notice that 
\begin{equation*}
\begin{split}
{\rm Tr}\left(Z^{(j_{1})} \rho_{k}^{(1)} \right) & = {\rm Tr} \left( (Z^{(j_{1})} \otimes I) (\rho_{k}^{(1)} \otimes \rho_{k}^{(2)}) \right), \\
{\rm Tr}\left( Z^{(j_{2})} \rho_{k}^{(2)} \right) &= {\rm Tr} \left(( I \otimes Z^{(j_{2})}) (\rho_{k}^{(1)} \otimes \rho_{k}^{(2)}) \right),
\end{split}
\end{equation*}
where $I$ is the identity operator. Therefore, we can relabel the qubit for the combined system described by the density operator $\rho_{k}^{(1)} \otimes \rho_{k}^{(2)}$ by $j$, running from $j = 1$ to $j = n_{1} + n_{2}$. Using this notation, the above expectations can be re-expressed as
\begin{equation*}
\begin{split}
{\rm Tr}\left( Z^{(j_{1})} \rho_{k}^{(1)}  \right) & = {\rm Tr}\left(Z^{(j)} \rho_{k}^{(1)} \otimes \rho_{k}^{(2)} \right), \quad j = j_{1}\\
{\rm Tr}\left( Z^{(j_{2})} \rho_{k}^{(2)} \right) & = {\rm Tr}\left( Z^{(j)} \rho_{k}^{(1)} \otimes \rho_{k}^{(2)}  \right), \quad j = n_{1} + j_{2}.
\end{split}
\end{equation*}
Following this idea, write out the outputs of two systems as follows,
\begin{eqnarray*}
\bar{y}_{k}^{(1)} & = &  C_{1} + \sum_{d_1 = 1 }^{R_1} \sum_{i_1 = 1}^{n_1} \cdots \sum_{i_{n_1} = i_{n_1 - 1} + 1 }^{n_1} \sum_{r_{i_1} + \cdots +r_{i_{n_1}} = d_1} w_{i_1, \ldots, i_{n_1}}^{r_{i_1}, \ldots, r_{i_{n_1}}} \langle Z^{(i_1)} \rangle_{k}^{r_{i_1}} \cdots \langle Z^{(i_{n_1})} \rangle_{k}^{r_{i_{n_1}}}, \\
\bar{y}_{k}^{(2)} & = &  C_{2} + \sum_{d_2 = 1 }^{R_2} \sum_{j_1 = 1}^{n_2} \cdots \sum_{j_{n_2} = j_{n_2 - 1} + 1 }^{n_2} \sum_{r_{j_1} + \cdots +r_{j_{n_2}} = d_2} w_{j_1, \ldots, j_{n_2}}^{r_{j_1}, \ldots, r_{j_{n_2}}} \langle Z^{(j_1)} \rangle_{k}^{r_{j_1}} \cdots \langle Z^{(j_{n_2})} \rangle_{k}^{r_{j_{n_2}}}.
\end{eqnarray*}
For any $\lambda \in \mathbb{R}$, let $n = n_{1} + n_{2}$ and $k$ denote the qubit index of the combined system running from $k=1$ to $k=n$, and $R = \max\{R_1, R_2\}$, then
\begin{eqnarray*}
\bar{y}_{k}^{(1)} + \lambda \bar{y}_{k}^{(2)} & = & C_{1} + \lambda C_{2} + \sum_{d=1}^{R} \sum_{k_{1} = 1}^{n} \cdots \sum_{k_{n} = k_{n-1} + 1}^{n} \sum_{r_{k_{1}} + \cdots + r_{k_n} = d} \bar{w}_{k_1, \ldots, k_n}^{r_{k_1}, \ldots, r_{k_n}} \langle Z^{(k_1)} \rangle_{k}^{r_{k_1}} \cdots \langle Z^{(k_n)} \rangle_{k}^{r_{k_n}},
\end{eqnarray*}
where the weights $\bar{w}_{k_1, \ldots, k_n}^{r_{k_1}, \ldots, r_{k_n}}$ are changed accordingly. For instance, if all $k_{m} \leq n_{1}$ for $m = 1, 2, \ldots, n$, then $\bar{w}_{k_1, \ldots, k_n}^{r_{k_1}, \ldots, r_{k_n}}  = w_{i_1, \ldots, i_{n_1}}^{r_{i_1}, \ldots, r_{i_{n_1}}}$, corresponding to the weights for the output $\bar{y}^{(1)}_{k}$. Similarly, let $R = R_1 + R_2$, 
\begin{eqnarray*}
\bar{y}_{k}^{(1)} \bar{y}_{k}^{(2)} & = & C_{1} C_{2} + \sum_{d=1}^{R} \sum_{k_{1} = 1}^{n} \cdots \sum_{k_{n} = k_{n-1} + 1}^{n} \sum_{r_{k_{1}} + \cdots + r_{k_n} = d} \hat{w}_{k_1, \ldots, k_n}^{r_{k_1}, \ldots, r_{k_n}} \langle Z^{(k_1)} \rangle_{k}^{r_{k_1}} \cdots \langle Z^{(k_n)} \rangle_{k}^{r_{k_n}}.
\end{eqnarray*}
Therefore, $\bar{y}_{k}^{(1)} + \lambda \bar{y}_{k}^{(2)}$ and $\bar{y}_{k}^{(1)} \bar{y}_{k}^{(2)}$ again have the same form as the right hand side of Eq.~\eqref{eq:mv-out} in the main text. This implies that for any functionals $F^{T^{(1)}}, F^{T^{(2)}} \in \mathcal{F}$, $F^{T^{(1)}} + \lambda F^{T^{(2)}} \in \mathcal{F}$ and $F^{T^{(1)}} F^{T^{(2)}} \in \mathcal{F}$. Thus, $\mathcal{F}$ forms a polynomial algebra.

It remains to show that for all $u_k \in D \cap [-L, L]$, $\| T(u_k) \rvert_{H_0(2^{n})} \|_{2-2} = \| (T^{(1)}(u_k) \otimes T^{(2)}(u_k))\rvert_{H_0(2^n)} \|_{2-2} \leq 1-\epsilon$ for some $0 < \epsilon \leq 1$. This will imply that $F^{T^{(1)}} + \lambda F^{T^{(2)}}$ and $F^{T^{(1)}} F^{T^{(2)}}$ are $w$-fading memory by Lemma~\ref{lemma:FMP}, and that $\mathcal{F}$ forms a polynomial algebra consisting of $w$-fading memory functionals. Suppose that for all $u_k \in D \cap [-L, L]$, $\|T(u_{k})\rvert_{H_0(2^{n_m})}\|_{2-2} \leq 1-\epsilon_{m}$ for $m=1, 2$. Adopting the proof of \cite[Proposition 3]{kubrusly2006concise}, let $A = \sum_{i} A_{i} \otimes \tilde{A}_{i}$ be a traceless Hermitian operator. Without loss of generality, we assume that $\{\tilde{A}_i\}$ is an orthonormal set with respect to the Hilbert-Schmidt inner product. Then $\{A_{i} \otimes \tilde{A}_{i}\}$ and $\{T^{(1)}(u_k)\rvert_{H_0(2^{n_{1}})} A_{i} \otimes \tilde{A}_{i} \}$ are two orthogonal sets. By the Pythagoras theorem, $T^{(1)}(u_k)\rvert_{H_0(2^{n_{1}})} \otimes I$ on the hyperplane of traceless Hermitian operators satisfies

\begin{equation*}
\begin{split}
& \| (T^{(1)}(u_k)\rvert_{H_0(2^{n_{1}})} \otimes I ) \sum_{i} A_{i} \otimes \tilde{A}_{i} \|_{2}^2  =\sum_{i} \| T^{(1)}(u_k)\rvert_{H_{0}(2^{n_{1}})} A_{i} \otimes \tilde{A}_i \|_{2}^2 \\
& = \sum_{i} \| T^{(1)}(u_k)\rvert_{H_{0}(2^{n_{1}})}A_{i} \|_2^{2} \|\tilde{A}_i \|_2^2 \leq \|T^{(1)}(u_k)\rvert_{H_{0}(2^{n_{1}})} \|_{2-2}^2 \sum_{i} \|A_i \|_2^2 \| \tilde{A}_i \|_2^2 \\
& = \|T^{(1)}(u_k)\rvert_{H_{0}(2^{n_{1}})} \|_{2-2}^2  \sum_{i}\| A_i\otimes \tilde{A}_i\|
_2^2 = \|T^{(1)}(u_k)\rvert_{H_{0}(2^{n_{1}})} \|_{2-2}^2  \| \sum_{i} A_i\otimes \tilde{A}_i\|_2^2 .
\end{split}
\end{equation*}
Therefore, $\|T^{(1)}(u_k)\rvert_{H_0(2^{n_1})} \otimes I \|_{2-2} \leq \|T^{(1)}(u_k)\rvert_{H_{0}(2^{n_1})}\|_{2-2}$. Similarly, a symmetric 
argument shows that $\|I \otimes T^{(2)}(u_k)\rvert_{H_0(2^{n_2})}\|_{2-2} \leq \|T^{(2)}(u_k)\rvert_{H_{0}(2^{n_2})}\|_{2-2}$. Therefore, when restricted to traceless Hermitian operators,
\begin{equation*}
\begin{split}
& \|(T^{(1)}(u_k) \otimes T^{(2)}(u_k))\rvert_{H_0(2^{n})}\|_{2-2} = \|(T^{(1)}(u_k)\rvert_{H_0(2^{n_1})} \otimes I) (I \otimes T^{(2)}(u_k)\rvert_{H_0(2^{n_2})})\|_{2-2} \\
&\leq \|T^{(1)}(u_k)\rvert_{H_0(2^{n_1})} \otimes I\|_{2-2} \|I \otimes T^{(2)}(u_k)\rvert_{H_0(2^{n_2})} \|_{2-2} \\
&\leq \|T^{(1)}(u_k)\rvert_{H_{0}(2^{n_1})}\|_{2-2} \|T^{(2)}(u_k)\rvert_{H_{0}(2^{n_2})}\|_{2-2} \leq (1-\epsilon_{1})(1-\epsilon_{2}).
\end{split}
\end{equation*}
The convergence of $T$ follows from Theorem~\ref{theorem:convergence}.
\end{proof}

\subsection{A universal class} \label{app:universal_class}

We now prove the universality of the class of dissipative quantum systems introduced in the main text. Recall that this class consists of $N$ non-interacting quantum subsystems initialized in a product state of the $N$ subsystems, where the dynamics of subsystem $K$ with $n_K$ qubits is governed by the CPTP map:
\begin{equation}\label{eqn:ES-CPTP}
T_{K}(u_k) \rho_{k-1}^{K} = {\rm Tr}_{i_{0}^{K}} (e^{-i H_{K} \tau} \rho_{k-1}^{K} \otimes \rho_{i_0, k}^{K} e^{i H_{k} \tau}),
\end{equation}
where 
$$ 
\rho_{i_{0}, k}^{K} = u_{k} |0 \rangle \langle 0| + (1-u_{k}) |1 \rangle \langle 1|, \quad  0 \leq u_{k} \leq 1
$$
\begin{equation}\label{eqn:exp-Ham}
H_{K} = \sum_{j_1 = 0}^{n_{K}} \sum_{j_2 = j_1 + 1}^{n_K} J_{K}^{j_1, j_2} (X^{(i^K_{j_{1}})} X^{(i^K_{j_{2}})} + Y^{(i^K_{j_{1}})} Y^{(i^K_{j_{2}})}) + \sum_{j=0}^{n_{K}} \alpha Z^{(i_j^K)},
\end{equation}
with $J^{j_{1}, j_{2}}_{K}$ and $\alpha$ being real-valued constants and ${\rm Tr}_{i_{0}^{K}}$ denoting the partial trace over the ancilla qubit. Let $\overline{H}_{K} = I \otimes \cdots \otimes H_{K} \otimes \cdots \otimes I$ with $H_{K}$ in the $K$-th position, the total Hamiltonian of $N$ subsystems is 
\begin{equation}
\label{eq:ES-hamiltonian}
\begin{split}
H & = \sum_{K = 1}^{N} \overline{H}_{K}.
\end{split}
\end{equation}
Writing $\rho_{k} = \bigotimes_{K=1}^{N} \rho_{k}^{K}$, the overall dynamics and the output are given by
\begin{equation}
\label{eq:ES-overall}
\begin{cases}
\rho_{k}  = T(u_{k}) \rho_{k-1} = \bigotimes_{K=1}^{N} T_{K}(u_{k}) \rho^{K}_{k-1}\\
\bar{y}_{k}  = h(\rho_{k}),
\end{cases}
\end{equation}
where $h$ is the multivariate polynomial defined by the right hand side of Eq.~\eqref{eq:mv-out} in the main text.

\begin{proposition}
\label{prop:universality_class}
Let $\mathcal{M}_{S}$ be the set of filters induced from dissipative quantum systems described by Eq.~\eqref{eq:ES-overall} such that each $T_{K}$ $(K = 1, \ldots, N)$ satisfies conditions in Theorem \ref{theorem:convergence}. Then for any null sequence $w$, the corresponding family of functionals $\mathcal{F}_{S}$ is dense in $C(K_{1}^{-}([0, 1]), \| \cdot \|_{w})$.
\end{proposition}

\begin{proof}
We first show $T_{K}(x)$ satisfies the conditions in Lemma~\ref{lemma:FMP} for all $x \in [0, 1]$. Let $x, y \in [0, 1]$ and $Z$ be the Pauli Z operator. By definition,
\begin{equation*}
\begin{split}
\| T_{K} (x) - T_{K}(y) \|_{2-2}  & = \sup_{\substack{A \in \mathcal{L}(\mathbb{C}^{2^n}) \\ \|A \|_{2} = 1}} \| (T_{K}(x) - T_{K}(y)) A \|_{2} \\
& = \sup_{\substack{A \in \mathcal{L}(\mathbb{C}^{2^{n}}) \\ \|A\|_{2} = 1}} \| {\rm Tr}_{i_{0}}^{K}(e^{-i H_{K} \tau} A \otimes (x-y) Z e^{i H_{K} \tau}) \|_{2} \\
& = |x-y| \sup_{\substack{A \in \mathcal{L}(\mathbb{C}^{2^{n}}) \\ \|A \|_{2} = 1}} \| {\rm Tr}_{i_{0}}^{K}(e^{-i H_{K} \tau} A \otimes  Z e^{i H_{K} \tau}) \|_{2} \\
& = |x - y| \| \tilde{T} \|_{2-2},
\end{split}
\end{equation*}
where $\tilde{T}$ is an input-independent CPTP map.

Now, the same argument in the proof of Lemma \ref{lemma:algebra} shows that $T = T_{1} \otimes \cdots \otimes T_{N}$ is convergent given the assumptions on each $T_K$. Furthermore, given two convergent systems whose dynamics are described by Eq.~\eqref{eq:ES-overall} with Hamiltonians $H^{(1)}$ and $H^{(2)}$, the total Hamiltonian of the combined system is $H = H^{(1)} \otimes I + I \otimes H^{(2)}$, which again has the form Eq.~\eqref{eq:ES-hamiltonian}. Therefore, by the above observation and Lemma \ref{lemma:algebra}, $\mathcal{F}_{S}$ forms a polynomial algebra, consisting of $w$-fading memory functionals for any null sequence $w$.

It remains to show $\mathcal{F}_{S}$ contains constants and separates points. Constants can be obtained by setting the weights $w_{i_1, \ldots, i_{n}}^{r_{i_1}, \ldots, r_{i_n}}$ in the output to be zero. To show the family $\mathcal{F}_{S}$ separates points, we state the following lemma for later use, whose proof can be found in \cite[Theorem 3.2]{lang1985complex}. 
\begin{lemma}
\label{lemma:power_series}
Let $f(\theta) = \sum_{n=0}^{\infty} x_{n} \theta^{n}$ be a non-constant real power series, having a non-zero radius of convergence. If $f(0) = 0$, then there exists $\beta>0$ such that $f(\theta) \neq 0$ for all $\theta$ with $\left|\theta \right| \leq \beta$ and $\theta \neq 0$.
\end{lemma}
Consider a single-qubit system interacting with a single ancilla qubit whose dynamics is governed by  Eq.~\eqref{eq:ES-overall}. Order an orthogonal basis of $\mathcal{L}(\mathbb{C}^{2})$ as $\mathcal{B} = \{I, Z, X, Y\}$.
Recall that the normal representations of a CPTP map $T$ and a density operator $\rho$ are given by
\begin{equation*}
\overline{T}_{i,j} = \frac{{\rm Tr}\left(B_{i} T(B_{j}) \right)}{2} \qquad \text{and} \qquad 
\overline{\rho}_{i} = \frac{{\rm Tr}(\rho B_{i})}{2},
\end{equation*}
where $B_{i} \in \mathcal{B}$. Without loss of generality, let $\tau = 1$ and set $J_{1}^{j_1, j_2} = J \in \mathbb{R}$ for all $j_{1}, j_{2}$ in the Hamiltonian given by Eq.~\eqref{eqn:exp-Ham}. We obtain the normal representation of the CPTP map defined in Eq.~\eqref{eqn:ES-CPTP} as
$$
\overline{T}(u_{k}) = \begin{pmatrix}
1 & 0 & 0 & 0 \\
\sin^2(2J)(2u_{k}-1) & \cos^{2}(2J) & 0 & 0\\
0 & 0 & \cos(2J)\cos(2\alpha) & -\cos(2J)\sin(2\alpha) \\
0 & 0 & \cos(2J)\sin(2\alpha) & \cos(2J)\cos(2\alpha)
\end{pmatrix}.
$$
When restricted to the hyperplane of traceless Hermitian operators, 
\begin{equation*}
\overline{T}\rvert_{H_{0}(2)} = \begin{pmatrix}
\cos^{2}(2J) & 0 & 0\\
0 & \cos(2J)\cos(2\alpha) & -\cos(2J)\sin(2\alpha) \\
0 & \cos(2J)\sin(2\alpha) & \cos(2J)\cos(2\alpha)
\end{pmatrix}
\end{equation*}
with $\left\| \overline{T}\rvert_{H_{0}(2)} \right\|_{2-2} = \sigma_{\max}(\overline{T}\rvert_{H_{0}(2)}) = |\cos(2J)|$. Here, $\left\| \cdot \right\|_{2-2}$ is the matrix $2$-norm and $\sigma_{\max}(\cdot)$ is the maximum singular value. Choose $J \neq \frac{z \pi}{2}$ for $z \in \mathbb{Z}$, then $|\cos(2J)| \leq 1-\epsilon$ for some $0 < \epsilon \leq 1$. By Theorem \ref{theorem:convergence}, $T$ is convergent and we choose an arbitrary initial density operator $\overline{\rho}_{-\infty} = \begin{pmatrix} 1/2 & 1/2 & 0 & 0 \end{pmatrix}^{T}$, corresponding to $\rho_{-\infty} = |0 \rangle \langle 0|$. If we only take the expectation $\langle Z \rangle$ in the output Eq.~\eqref{eq:mv-out} by setting the degree $R=1$, then this single-qubit dissipative quantum system induces a functional
\begin{equation*}
\begin{split}
F^{T}(u) = w \left[ \left( \overrightarrow{\prod}_{j=0}^{\infty} \overline{T}(u_{-j}) \right) \overline{\rho}_{-\infty} \right]_{2} + C,
\end{split}
\end{equation*}
for all $u \in K^{-}_{1}([0, 1])$. Here, $[\cdot]_{2}$ refers to the second element of the vector, corresponding to $\langle Z \rangle$ given the order of the orthogonal basis elements in $\mathcal{B}$. Given two input sequences $u \neq v$ in $K^{-}_{1}([0, 1])$, consider two cases:\\
(i) If $u_{0} \neq v_{0}$, choose $J = \frac{\pi}{4}$ such that $\cos^2(2J) = 0$ and $\sin^{2}(2J)=1$. Then
$$
F^{T}(u) - F^{T}(v) = w(u_{0} - v_{0}) \neq 0.
$$
(ii)If $u_{0} = v_{0}$, 
$$
F^{T}(u) - F^{T}(v) = w \sin^2(2J) \sum_{j=0}^{\infty} \left(\cos^2(2J) \right)^{j}(u_{-j} - v_{-j}).
$$
Let $\theta = \cos^{2}(2J)$, then given our choice of $J$, $0 \leq \theta \leq 1-\epsilon$ and $\sin^2(2J) \geq \epsilon$ for some $0< \epsilon \leq 1$. Consider the power series
$$ 
f(\theta) = \sum_{j=0}^{\infty} \theta^{j} (u_{-j} - v_{-j}),
$$
since $\left| u_{-j} - v_{-j} \right| \leq 1$, $f(\theta)$ has a non-zero radius of convergence $R$ such that $(-1, 1) \subseteq R$. Moreover, $f(\theta)$ is non-constant and $f(0) = 0$. The separation of points follows from invoking Lemma \ref{lemma:power_series}.

Finally, the universality property of $\mathcal{F}_{S}$ follows from Theorem~\ref{theorem:universal_dissipative}.
\end{proof}

\subsection{Detailed numerical experiment settings} \label{app:numerical}
In this section, we describe detailed formulas for the NARMA tasks, simulation of decoherence and experimental conditions for ESNs and the Volterra series.

\subsubsection{The NARMA task}
The general $m$th-order NARMA I/O map is described as \cite{atiya2000new}:
$$
y_{k} = 0.3 y_{k-1} + 0.05 y_{k-1} \left( \sum_{j=0}^{\tau_{\rm NARMA}-1} y_{k-j-1} \right) + 1.5 u_{k-\tau_{\rm NARMA}} u_{k} + \gamma.
$$
where $\gamma \in \mathbb{R}$. In the main text, we consider $\tau_{\rm NARMA} = \{15, 20, 30, 40\}$. For $\tau_{\rm NARMA} = \{15, 20\}$, we set $\gamma = 0.1$. For $\tau_{\rm NARMA} = \{30, 40\}$, $\gamma$ is set to be $0.05$ and $0.04$ respectively. A random input sequence $u^{(r)}$, where each $u_{k}^{(r)}$ is randomly uniformly chosen from $[0, 0.2]$, is deployed for all the computational tasks. This range is chosen to ensure stability of the NARMA tasks.

\subsubsection{Decoherence}
\label{app:decoherence} 
We consider the dephasing, decaying and generalized amplitude damping (GAD) noise, which are of experimental importance. The dephasing noise has the Kraus operators \cite{nielsen2002quantum}:
$$
M_{0} = \sqrt{\frac{1 + \sqrt{1-p}}{2}} I, M_{1} = \sqrt{\frac{1 - \sqrt{1-p}}{2}} Z,
$$
where $\sqrt{1-p} = e^{-2 \frac{\gamma}{S}\delta_{t}}$. Therefore, we implement single-qubit phase-flip for all $n$ system and ancilla qubits. That is for $j=1,\ldots, n+1$ the density operator $\rho$ for the system and ancilla qubits undergoes the evolution:
$$
\rho \rightarrow \frac{1+e^{-2 \frac{\gamma}{S} \delta_{t}}}{2} \rho + \frac{1-e^{-2 \frac{\gamma}{S} \delta_{t}}}{2} Z^{(j)} \rho Z^{(j)},$$
where $Z^{(j)}$ denotes the Pauli $Z$ operator for qubit $j$.

The generalized amplitude damping (GAD) channel captures the effect of dissipation to an environment at a finite temperature $\lambda \in [0, 1]$. Its Kraus operators are defined by
\begin{equation*} 
\begin{split}
& M_{0} = \sqrt{\lambda}\begin{pmatrix}
1 & 0 \\ 0 & \sqrt{1-p} \end{pmatrix}, M_{2} = \sqrt{\lambda}\begin{pmatrix}0 & \sqrt{p}\\0 & 0 \end{pmatrix}, \\
& M_{3} = \sqrt{1-\lambda} \begin{pmatrix} \sqrt{1-p} & 0 \\ 0 & 1 \end{pmatrix}, M_{4} = \sqrt{1-\lambda} \begin{pmatrix} 0 & 0 \\\sqrt{p} & 0\end{pmatrix}.
\end{split}
\end{equation*}
When $\lambda=1$, the GAD channel corresponds to the amplitude damping channel (decaying noise). We simulate the generalized amplitude damping channel for $\lambda=\{0.2, 0.4, 0.6, 0.8\}$. To implement the GAD channel with the same noise strengths as the dephasing channel, we set $\sqrt{1-p} = e^{-2 \frac{\gamma}{S} \delta_{t}}, \sqrt{p} = \sqrt{1-e^{-4 \frac{\gamma}{S} \delta_t}}$ to be the same as the dephasing noise.

Following the discussion in Sec.~\ref{subsec:SA_decoherence}, Fig.~\ref{figure:GAD_NMSE} plots the average SA NMSE for the LRPO, Missile, NARMA15 and NARMA20 tasks under the GAD channel for all the chosen temperature parameters. Fig.~\ref{figure:final_steps_dd} and Fig.~\ref{figure:final_steps_GAD} plot the average sum of modulus of off-diagonal elements in the system density operator, for the last 50 timesteps of the SA samples, under all noise types discussed above.

\begin{figure}[htb!]
\centering
\includegraphics[trim={35mm 15mm 35mm 10mm}, clip, scale=0.42]{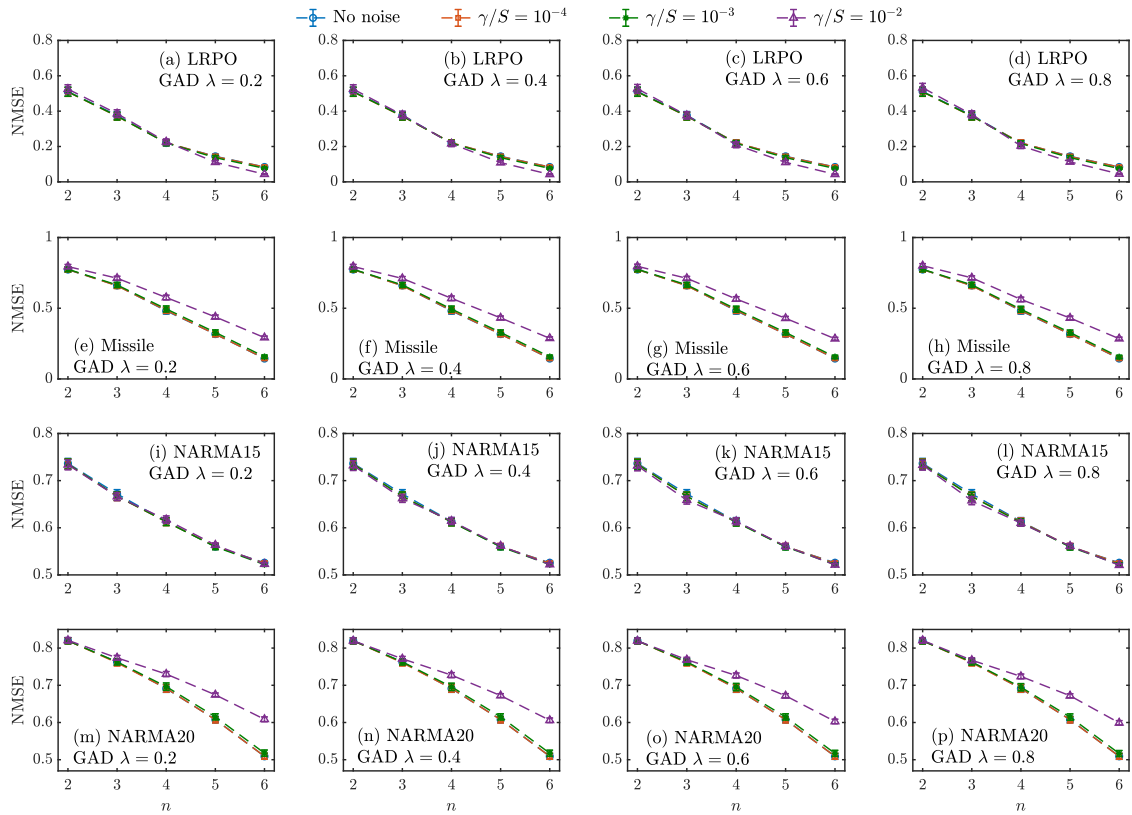}
\caption{Average SA NMSE for the LRPO, Missile, NARMA15 and NARMA20 tasks under GAD for $\lambda=\{0.2, 0.4, 0.6, 0.8\}$}
\label{figure:GAD_NMSE}
\end{figure}

\begin{figure}[htb!]
\centering
\includegraphics[trim={25mm 5mm 10mm 10mm}, clip, scale=0.4]{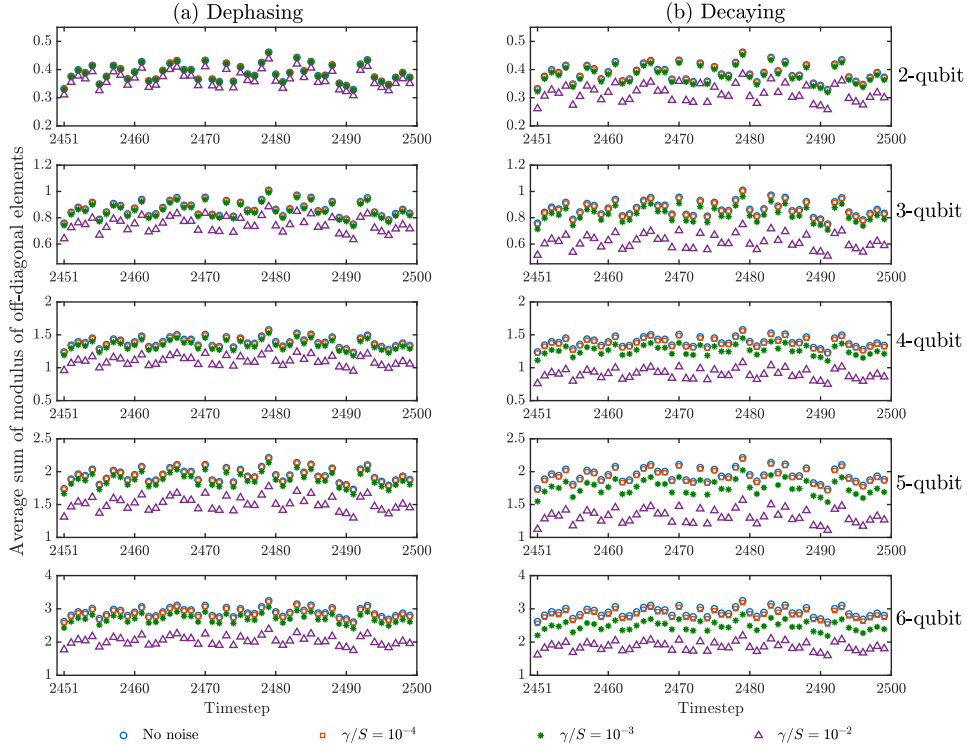}
\caption{Average sum of modulus of off-diagonal elements in the density operator, for the last 50 timesteps of the SA samples, under the (a) dephasing noise and (b) decaying noise}
\label{figure:final_steps_dd}
\end{figure}

\begin{figure}[htb!]
\centering
\includegraphics[trim={35mm 5mm 15mm 10mm}, clip, scale=0.4]{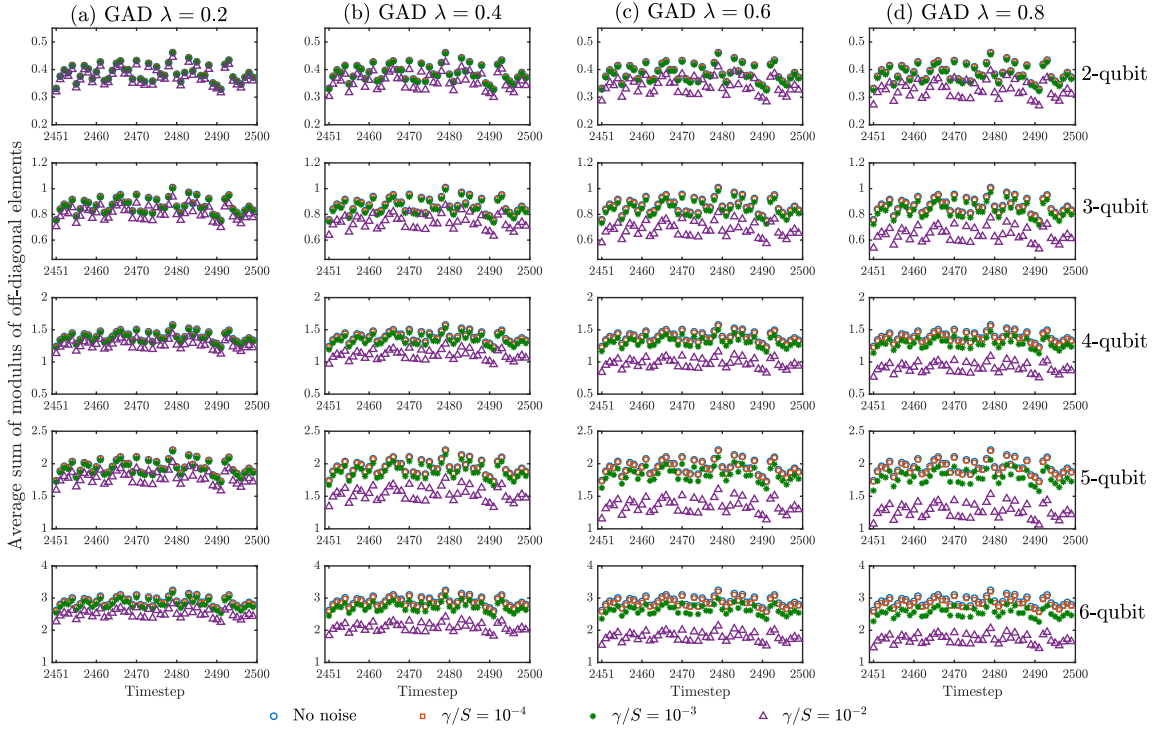}
\caption{Average sum of modulus of off-diagonal elements in the density operator, for the last 50 timesteps of the SA samples, under GAD for (a) $\lambda=0.2$, (b) $\lambda=0.4$, (c) $\lambda=0.6$ and (d) $\lambda=0.8$}
\label{figure:final_steps_GAD}
\end{figure}

\subsubsection{The echo state networks}
An ESN with $m$ reservoir nodes is a type of recurrent neural network with a $m \times 1$ input matrix $W_{i}$, a $m \times m$ reservoir matrix $W_{r}$ and an $1 \times m$ output matrix $W_{o}$.  The state evolution and output are given by \cite{JH04}
\begin{equation*}
\begin{cases}
x_{k} = \tanh(W_{r}x_{k-1} + W_{i}u_{k}) \\
\hat{y}_{k} = W_{o} x_k + w_{c}, 
\end{cases}
\end{equation*}
where $w_{c}$ is a tunable constant and $\tanh(\cdot)$ is an element-wise operation.

In the numerical examples, lengths of washout, learning and evaluation phases for ESNs and SA are the same. Given an output sequence $y$ to be learned, the output weights $w_{c}$ and $W_{o}$ are optimized via standard least squares to minimize $\sum_{k}|y_{k} - \hat{y}_{k}|^2$, for timesteps $k$ during the training phase. We now detail the experimental conditions for ESNs in various subsections of the numerical experiments (Sec.~\ref{sec:numerical}).

For the comparison given in Subsection~\ref{subsec:overview}, we set the reservoir size to be $m \in \mathcal{M}=\{10, 20, 30, 40, 50, 100, 150, 200, 250, 300, 400, 500, 600, 700, 800\}$. Here, the number of computational nodes is $m+1$ for each $m$. For each computational task and each $m$, the average NMSE of 100 ESNs is reported. The average NMSE for ESNs is obtained as follows. For each reservoir size $m$, we prepare 100 ESNs with elements of $W_{r}$ randomly uniformly chosen $[-2, 2]$. Let $\mathcal{S}$ denote the set of 10 points evenly spaced between $[0.01, 0.99]$. For each of the 100 ESNs, we scale the maximum singular value of $W_{r}$ to $\sigma_{\max}(W_{r})=s$ for all $s \in \mathcal{S}$. This ensures the convergence and fading memory property of ESNs \cite{grigoryeva2018echo}. For each of the chosen $s$, the elements of $W_{i}$ are randomly uniformly chosen within $[-\delta, \delta]$, where $\delta$ is chosen from the set $\mathcal{I}$ of $10$ points evenly spaced between $[0.01, 1]$. Now, for the $i$-th ($i=1, \ldots, 100$) ESN with parameter $(m, s, \delta)$, we denote its associated NMSE to be ${\rm NMSE}_{(m, s, \delta, i)}$. For each reservoir size $m$, the average NMSE is computed as $\frac{1}{\mathcal{|S|}} \frac{1}{\mathcal{|I|}} \frac{1}{100} \sum_{s \in \mathcal{S}}\sum_{\delta \in \mathcal{I}} \sum_{i=1}^{100}{\rm NMSE}_{(m, s, \delta, i)}$. Fig.~\ref{figure:ESN_NMSE_NARMA} summarizes the average ESNs NMSE for the LRPO, Missile, NARMA15 and NARMA20 tasks.

% \begin{table}
% \scriptsize
% \begin{tabular}{c c c c c}
% \hline\noalign{\smallskip}
% \noalign{\vskip 0.5mm}
%  & \multicolumn{2}{l}{NARMA15} & \multicolumn{2}{l}{NARMA20} \\
%  \hline
%  \noalign{\vskip 0.5mm}
% Reservoir size & Average NMSE & Standard error  & Average NMSE & Standard error \\
%  & ($\times 10^{-3}$)  & ($\times 10^{-3}$) & ($\times 10^{-3}$)  & ($\times 10^{-3}$)   \\
% \hline
%  10    & 3.4   & 0.023 & 4.0   &  0.046 \\
%  20    & 3.2   & 0.048 & 4.0   &  0.063 \\
%  30    & 3.0   & 0.038 & 3.9   &  0.065 \\
%  40    & 2.9   & 0.032 & 3.8   &  0.055 \\
%  50    & 2.8   & 0.045 & 3.7   &  0.049 \\
%  100   & 2.6   & 0.062 & 3.6   &  0.029 \\
%  150   & 2.5   & 0.038 & 3.5   &  0.025 \\
%  200   & 2.4   & 0.029 & 3.4   &  0.034 \\
%  250   & 2.3   & 0.021 & 3.4   &  0.046 \\
%  300   & 2.3   & 0.021 & 3.3   &  0.041 \\
%  400   & 2.2   & 0.014 & 3.3   &  0.042 \\
%  500   & 2.2   & 0.012 & 3.3   &  0.038 \\
%  600   & 2.2   & 0.014 & 3.2   &  0.028 \\
%  700   & 2.1   & 0.017 & 3.2   &  0.027 \\
%  800   & 2.1   & 0.016  & 3.2   &  0.029 \\
% \noalign{\smallskip}\hline
% \end{tabular}
% \caption{Average ESNs NMSE for NARMA15 and NARMA20 tasks. Results are rounded to two significant figures.}
% \label{table:ESN_NARMA}
% \end{table}

\begin{figure}[!htb]
\includegraphics[trim={0 0 0 0}, clip, scale=0.3]{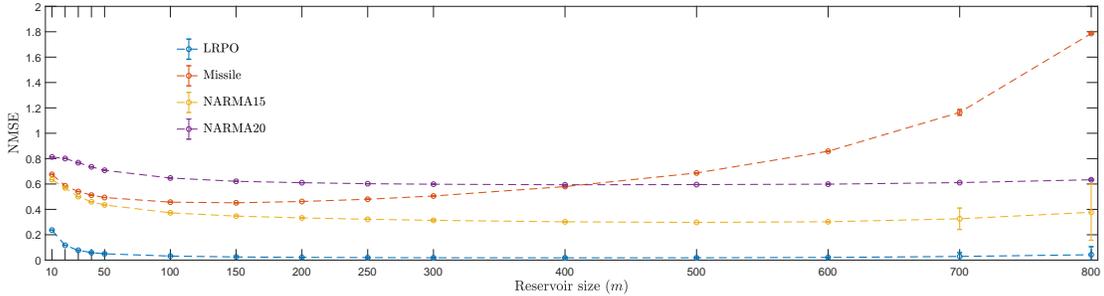}
\caption{Average NMSE of ESNs for the LRPO, Missile, NARMA15 and NARMA20 tasks. The data symbols obscure the error bars, which represent the standard error}
\label{figure:ESN_NMSE_NARMA}
\end{figure}

For the further comparison in Subsection~\ref{subsec:further_comparison}, ESNs are simulated to approximate the LRPO, Missile, NARMA15, NARMA20, NARMA30 and NARMA40 tasks. The reservoir size of ESNs for each task is set to be $m \in \mathcal{M} = \{256, 300, 400, 500\}$. For each $m$, the number of computational nodes $\mathcal{C}$ for ESNs is $$\mathcal{C} \in \mathcal{N}_{4} \cup \mathcal{N}_{5} \cup \mathcal{N}_{6} = \{5, 6, 7, 15, 21, 28, 35, 56, 70, 84, 126, 210, 252\},$$ where $\mathcal{N}_{n}$ denotes the chosen numbers of computational nodes for $n$-qubit SA defined as follows. Recall that in this experiment, 4-, 5- and 6-qubit SA with varying degrees $R$ in the output are chosen. For 4-qubit SA, $R_{4}=\{1, \ldots, 6\}$ correspond to the number of computational nodes $\mathcal{N}_{4} = \{5, 15, 35, 70, 126, 210\}$. For 5-qubit SA, $R_{5}=\{1, \ldots, 5\}$, such that $\mathcal{N}_{5} = \{6, 21, 56, 126, 252\}$. For 6-qubit SA, $R_{6}=\{1, \ldots, 4\}$, such that $\mathcal{N}_{6} = \{7, 28, 84, 210\}$. To compute the output weights $W_{o}$ and $w_{c}$ when $\mathcal{C} < m+1$, we first optimize $W_{o}$ and $w_{c}$ by standard least squares. Then choose $\mathcal{C}-1$ elements of $W_{o}$ with the largest absolute values and their corresponding elements $x_{k}'$ from the state $x_k$. These $\mathcal{C}-1$ state elements $x'_k$ are used to re-optimize $\mathcal{C}-1$ elements $W'_{o}$ of $W_{o}$ and $w'_{c}$ via standard least squares. At each timestep $k$, the full state $x_k$ evolves, while the output is computed as $\hat{y}' = W'_{o} x'_{k} + w'_{c}$. For this numerical experiment, the chosen parameters $\mathcal{S}$ and $\mathcal{I}$ of ESNs are the same as above. For the $i$-th ESN with parameter $(m, s, \delta)$, the number of computational nodes $\mathcal{C}$ varies. Let ${\rm NMSE}_{(m, \mathcal{C}, s, \delta, i)}$ denotes the corresponding NMSE. For each $m$ and each $\mathcal{C}$, we report the average NMSE computed as $\frac{1}{\mathcal{|S|}} \frac{1}{\mathcal{|I|}} \frac{1}{100} \sum_{s \in \mathcal{S}}\sum_{\delta \in \mathcal{I}} \sum_{i=1}^{100}{\rm NMSE}_{(m, \mathcal{C}, s, \delta, i)}$.

\subsubsection{The Volterra series}
The discrete-time finite Volterra series with kernel order $o$ and memory $p$ is given by \cite{BC85}
$$
\hat{y}_k = h_0 + \sum_{i=1}^{o} \sum_{j_1, \cdots, j_i = 0}^{p-1} h_{i}^{j_1, \cdots, j_i} \prod_{l = 1}^{i} u_{k-j_l},
$$
where $u_{k - j}$ is the delayed input, $h_{0}$ and $h_{i}^{j_1, \cdots, j_i}$ are real-valued kernel coefficients (or output weights in our context). Notice that when memory $p=1$, the Volterra series is a map from the current input $u_{k}$ to the output $\hat{y}_{k}$. The kernel coefficients are optimized via linear least squares to minimize $\sum_{k} |y_{k} - \hat{y}_{k}|^2$ during the training phase, where $y$ is the target output sequence to be learned.

The number of computational nodes, that is the number of kernel coefficients $h_{0}$ and $h_{i}^{j_1, \cdots, j_i}$, is given by $(p^{o+1} - p) / (p - 1) + 1$. We vary the parameters of the Volterra series as follows: for each $o = \{2, \ldots, 8 \}$, choose $p$ from $\{2, \ldots, 27\}$ such that the maximum number of computational nodes does not exceed 801. Note that for $o=1$, the output of the Volterra series is a linear function of delayed inputs. Since we are interested in nonlinear I/O maps, we choose $o \geq 2$. Table~\ref{table:Volterra_kp} summarizes the number of computational nodes as $o$ and $p$ vary. Fig.~\ref{figure:V_NMSE_o} shows the Volterra series NMSE according to the kernel order and memory.

\begin{table}[!htb]
\caption{Values of $o$ and $p$ for the Volterra series and the corresponding number of computational nodes. The empty entries indicate that for the chosen $o$ and $p$, the number of computational nodes exceeds 801}
\label{table:Volterra_kp}
\resizebox{\textwidth}{!}{
\scriptsize
\begin{tabular}{l l l l l l l l l l l l l l l l l l l l l l l l l l l l }
\hline\noalign{\smallskip}
\noalign{\vskip 0.5mm}
\backslashbox{$o$}{$p$} & 2 & 3 & 4 & 5 & 6 & 7 & 8 & 9 & 10 & 11 & 12 & 13 & 14 & 15 & 16 & 17 & 18 & 19 & 20 & 21 & 22 & 23 & 24 & 25 & 26 & 27 \\     
\hline
2 & 7 & 13 & 21 & 31 & 43 & 57 & 73 & 91 & 111 & 133 & 157 & 183 & 211 & 241 & 273 & 307 & 343 & 381 & 421 & 463 & 507 & 553 & 601 & 651 & 703 & 757\\
3 & 15 & 40 & 85 & 156 & 259 & 400 & 585  \\
4 & 31 & 121 & 341 & 781 \\
5 & 63 & 364 & & & & & & \\
6 & 127 & & & & & & & & \\
7 & 255 & & & & & & & & \\
8 & 511 & & & & & & & &  \\
\hline\noalign{\smallskip}
\noalign{\vskip 0.5mm}
\end{tabular}}

\end{table}

\begin{figure}[ht!]
\includegraphics[trim={25mm 10mm 25mm 0mm}, clip, scale=0.46]{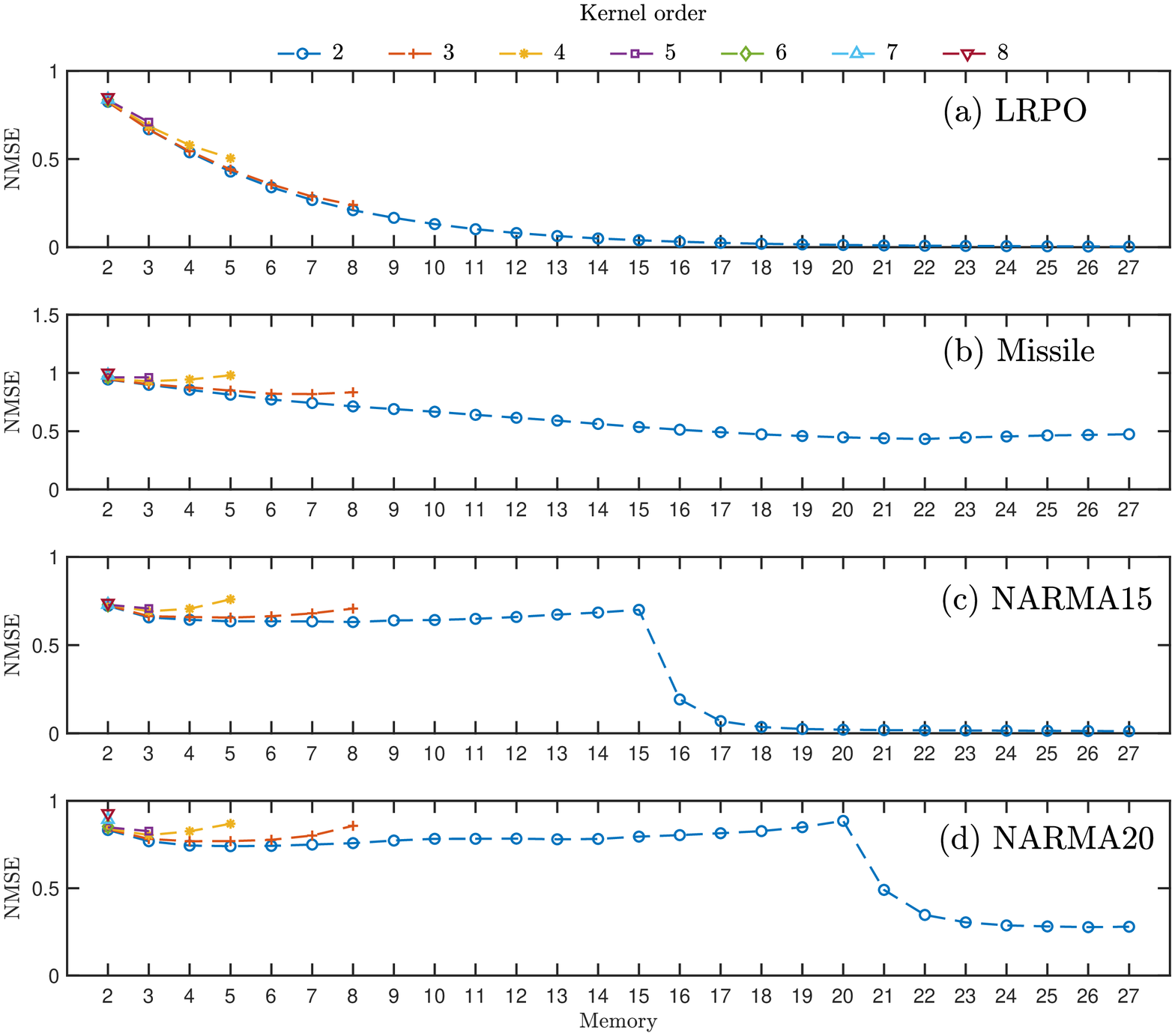}
\caption{NMSE of the Volterra series according to kernel order and memory, for the (a) LRPO, (b) Missile, (c) NARMA15 and (d) NARMA20 tasks}
\label{figure:V_NMSE_o}
\end{figure}

It is observed in Fig.~\ref{figure:V_NMSE_o} that as the kernel order increases, the Volterra series task performance does not improve. On the other hand, as the memory increases for kernel order 2, the Volterra series task performance improves. The improvement is particularly significant as the memory $p$ coincides with the delay for NARMA tasks, that is when $p = \tau_{\rm NARMA} + 1$.

\bibliographystyle{aipauth4-1}
\bibliography{qml}

\end{document}